\documentclass[10pt,numbers,squareay]{sigplanconf-pldi16}

\usepackage{alltt}
\usepackage{amsmath}
\usepackage{amssymb}
\usepackage{amsthm}
\usepackage{mathrsfs}
\usepackage{bbm}
\usepackage{dsfont}
\usepackage{esvect}
\usepackage{subfigure}
\usepackage[disable]{todonotes}
\usepackage{booktabs}
\usepackage{thmtools, thm-restate}
\usepackage[section]{algorithm}
\usepackage{algpseudocode}
\usepackage{multirow}
\usepackage{capt-of}
\usepackage[scaled]{helvet}
\usepackage{courier}
\usepackage{enumitem}
\usepackage[utf8]{inputenc}
\usepackage[T1]{fontenc}
\usepackage{microtype}
\usepackage{balance}
\usepackage{inconsolata}
\usepackage{fancyvrb}
\usepackage{rotating}
\usepackage{framed}
\usepackage{filecontents}

\newcommand{\cbstart}{}
\newcommand{\cbend}{}

\usepackage[hidelinks,breaklinks,draft=false]{hyperref}
\usepackage[all]{hypcap}

\hypersetup{colorlinks, linkcolor={red!50!black}, citecolor={blue!50!black}, urlcolor={blue!80!black}}
\newcommand{\aautoref}[2]{{\autoref{#1}}}

\usepackage{tikz,pgfplots,pgfplotstable}
   \usetikzlibrary{arrows,calc,shapes,automata,plotmarks,decorations,external}
\usepgfplotslibrary{groupplots}

\newcommand{\sem}[1]{{[\![{#1}]\!]}}

\def\bibfont{\small}

\usepackage{listings}
\lstnewenvironment{code}[1][]{
   \noindent
   \minipage{\linewidth} 
   \vspace{0.5\baselineskip}
   \lstset{
        language=C,
        basicstyle=\ttfamily\footnotesize,
        frame=tb,
        showstringspaces=false,
        xleftmargin=.05\textwidth, xrightmargin=.05\textwidth,
        #1}}
{\endminipage}

\usepackage{mathpartir}

\newcommand{\subtype}{{\;\sqsubseteq\;}}

\newcommand{\subsketch}{{\;\trianglelefteq\;}}

\newcommand{\proves}{{\;\vdash\;}}
\newcommand{\notproves}{{\not\;\vdash\;}}

\newcommand{\typename}[1]{\underline{\mathsf{#1}}}

\newcommand{\stPop}[1]{{\mathsf{pop}\; {#1}}}
\newcommand{\stPush}[1]{{\mathsf{push}\; {#1}}}

\newcommand{\one}{\underline{1}}
\newcommand{\zero}{\underline{0}}

\newcommand{\denotes}[1]{[\![{#1}]\!]}

\newcommand{\scheme}[3]{\forall{#1}.\,{#2}\!\Rightarrow\!{#3}}

\newcommand{\rewrite}{\longrightarrow}
\newcommand{\rewrites}{\stackrel{*}{\rewrite}}

\newcommand{\stStart}{{^\#\textsc{Start}}}
\newcommand{\stEnd}{{^\#\textsc{End}}}
\newcommand{\Deriv}[1]{\mathsf{Deriv}_{\mathcal{#1}}}

\newcommand{\term}[1]{\textsc{Var}\;{#1}}

\newcommand{\pdsrule}[4]{	\langle {#1}; {#2} \rangle \hookrightarrow
	\langle {#3}; {#4} \rangle}

\makeatletter
\providecommand{\bigsqcap}{  \mathop{    \mathpalette\@updown\bigsqcup
  }}
\newcommand*{\@updown}[2]{  \rotatebox[origin=c]{180}{$\m@th#1#2$}}
\makeatother

 \theoremstyle{definition}
 \newtheorem{example}{Example}[section]

 \theoremstyle{definition}
 \newtheorem{lemma}{Lemma}[section]

 \theoremstyle{definition}
 \newtheorem{theorem}{Theorem}[section]

 \theoremstyle{definition}
 \newtheorem{corollary}{Corollary}[section]

 \theoremstyle{definition}
 \newtheorem{definition}{Definition}[section]

 \theoremstyle{definition}
 \newtheorem{note}{Note}
 
\begin{document}
\toappear{}

\title{Polymorphic Type Inference for Machine Code\thanks{This research was developed with
funding from the Defense Advanced Research Projects Agency (DARPA). The views, opinions,
and/or findings contained in this material are those of the authors and should not be interpreted
as representing the official views or policies of the Department of Defense or the U.S. Government.\\
{\textsc{Distribution A.}} Approved for public release; distribution unlimited.}}

\authorinfo{Matthew Noonan \and Alexey Loginov \and David Cok}
{GrammaTech, Inc.\\Ithaca NY, USA}
{\{mnoonan,alexey,dcok\}@grammatech.com}
\titlebanner{DRAFT -- Do not distrubute}
\copyrightyear{2016}

\maketitle

\begin{abstract}
For many compiled languages, source-level types are erased very early
in the compilation process.  As a result, further compiler passes 
may convert type-safe source into type-unsafe machine code.
Type-unsafe idioms in the original source and type-unsafe optimizations
mean that type information in a stripped binary is essentially nonexistent.
The problem of recovering high-level types by performing type inference over
stripped machine code is called {\em type reconstruction}, and offers a useful
capability in support of reverse engineering and decompilation.

In this paper, we motivate and develop a novel type system and algorithm for
 machine-code type inference.  The features of this type system were
developed by surveying a wide collection of common source- and machine-code
idioms, building a catalog of challenging cases for type reconstruction.  We
found that these idioms place a sophisticated set of
requirements on the type system, inducing features such as recursively-constrained
polymorphic types.  Many of the features we identify are often seen only in
expressive and powerful type systems used by high-level functional languages.

Using these type-system features as a guideline, we have developed Retypd:
a novel static type-inference algorithm for
machine code that supports recursive types, polymorphism, and subtyping.
Retypd yields more accurate inferred types than existing algorithms,
while also enabling new capabilities such as reconstruction of
pointer {\verb|const|} annotations with 98\% recall.
Retypd can operate on weaker
program representations than the current state of the art, removing the
need for high-quality points-to information that may be impractical
to compute.
 \end{abstract}

\category{F.3.2}{Logics and Meanings of Programs}{Semantics of Programming Languages}
\category{D.2.7}{Software Engineering}{Distribution, Maintenance, and Enhancement}
\category{D.3.3}{Programming Languages}{Language Constructs and Features}
\category{F.4.3}{Mathematical Logic and Formal Languages}{Formal Languages}
\keywords{} Reverse Engineering, Type Systems, Polymorphism, Static Analysis, Binary Analysis, Pushdown Automata

\section{Introduction}

In this paper we introduce Retypd, a machine-code type-inference tool that finds
{\bf re}gular {\bf ty}pes using {\bf p}ush{\bf d}own systems.  Retypd
includes several novel features targeted at improved types for
reverse engineering, decompilation, and high-level program analyses.
These features include:
\begin{itemize}[noitemsep]
\item Inference of most-general type schemes (\autoref{simplify})
\item Inference of recursive structure types (\autoref{type-recovery-example})
\item Sound analysis of pointer subtyping (\autoref{model-ptrs})
\item Tracking of customizable, high-level information such as purposes and
typedef names (\autoref{sketch-sec})
\item Inference of type qualifiers such as {\verb|const|} (\autoref{const-recovery-stats})
\item No dependence on high-quality points-to data (\autoref{sec-evaluation})
\item More accurate recovery of source-level types (\autoref{sec-evaluation})
\end{itemize}

Retypd continues in the tradition of SecondWrite \citep{2ndwrite} and TIE \citep{tie} by introducing a principled
static type-inference algorithm applicable to stripped binaries.  Diverging from previous work on machine-code type reconstruction, we use a rich type system that supports polymorphism, mutable references, and recursive types.  The principled type-inference phase is followed
by a second phase that uses heuristics to ``downgrade'' the inferred types to human-readable
C types before display.
By factoring type inference into two phases, we can sequester unsound heuristics and quirks of the C type systems
 from the sound core of the type-inference engine.  This adds a degree of freedom to
the design space so that we may leverage a relatively complex type system during 
type analysis, yet still emit familiar C types for the benefit of the reverse engineer.

Retypd operates on an intermediate representation (IR) recovered by automatically disassembling a
binary using GrammaTech's static analysis tool CodeSurfer\textregistered{} for Binaries \citep{csurf}.  By generating type constraints
from a TSL-based abstract interpreter \citep{lim12}, Retypd can operate uniformly on binaries for any platform supported by CodeSurfer, including x86, x86-64, and ARM.

During the development of Retypd, we carried out an extensive investigation of common machine-code
idioms in compiled C and C++ code that create challenges for existing type-inference methods.  For each challenging case, we
identified requirements for any type system that could correctly type the idiomatic code. The results
of this investigation appear in \autoref{challenge-sec}.  The type system used by
Retypd was specifically designed to satisfy these requirements.  These common idioms pushed us
into a far richer type system than we had first expected, including features like recursively constrained type schemes that have not previously been applied to machine-code type inference.

Due to space limitations, details of the proofs and algorithms appear in the appendices, which are available in the online version of this paper at \verb|arXiv:1603.05495| \citep{arxiv}. Scripts and data sets used for evaluation also appear there.

\section {Challenges}
\label{challenge-sec}

There are many challenges to carrying out type inference on machine code, and many common idioms
that lead to sophisticated demands on the feature set of a type system.  In this section, we
describe several of the challenges seen during the development of Retypd that led to our particular
combination of type-system features.

\subsection {Optimizations after Type Erasure}
Since type erasure typically happens early in the compilation process, many compiler optimizations
may take well-typed machine code and produce functionally equivalent but ill-typed results. We found
that there were three common optimization techniques that required special care: the use of a variable as a syntactic constant, early returns along error paths, and the re-use of stack slots.

\smallskip
\noindent
{\em Semi-syntactic constants}:
Suppose a function with signature {\verb|void f(int x, char* y)|} is invoked 
as {\verb|f(0, NULL)|}.  This will usually be compiled to x86 machine code similar to

\begin{code}
xor  eax, eax
push eax        ; y := NULL
push eax        ; x := 0
call f
\end{code}

\noindent
This represents a code-size optimization, since {\verb|push eax|} can be encoded in one byte instead of
the five bytes needed to push an immediate value (0).
We must be careful that the type variables for $x$ and $y$ are not unified; here, {\verb|eax|} is being
used more like a syntactic constant than a dynamic value that should be typed.

\smallskip
\noindent
{\em Fortuitous re-use of values}:
A related situation appears in the common control-flow pattern represented by the snippet of C and the corresponding machine code in \autoref{fortuitous}.
Note that on procedure exit, the return value in {\verb|eax|} may have come from either the return
value of {\verb|S2T|} or from the return value of {\verb|get_S|}
(if {\verb|NULL|}).  If this
situation is not detected, we will see a false relationship between the incompatible return types of
{\verb|get_T|} and {\verb|get_S|}.

\begin{figure}
\hrule
\vrule
\begin{minipage}{.48\linewidth}
\begin{code}[frame=none]
 T * get_T(void)
 {
     S * s = get_S();
     if (s == NULL) {
         return NULL;
     }
     T * t = S2T(s);
     return t;
 }
\end{code}
\end{minipage}
\vrule
\begin{minipage}{.48\linewidth}
\begin{code}[frame=none]
 get_T:
     call get_S
     test eax, eax
     jz   local_exit
     push eax
     call S2T
     add  esp, 4
 local_exit:
     ret
\end{code}
\end{minipage}
\hspace{2pt}
\vrule
\vspace{-4pt}
\caption{A common fortuitous re-use of a known value.}
\label{fortuitous}
\end{figure}

\smallskip
\noindent
{\em Re-use of stack slots}:
 If a function uses two variables of the same size in disjoint scopes,
there is no need to allocate two separate stack slots for those variables. Often the optimizer will reuse a
stack slot from a variable that has dropped out of scope.  This is true even if the new variable has a different
type.  This optimization even applies to the stack slots used to store formal-in parameters, as in \autoref{type-recovery-example}; once the function's argument is no
longer needed, the optimizer can overwrite it with a local variable of an incompatible type.

More generally, we cannot assume that the map from program variables to physical locations
is one-to-one.  We cannot even make the weaker assumption that the program variables inhabiting a
single physical location at different times will all belong to a single type.

\smallskip
We handle these issues through a combination of type-system features (subtyping instead of unification)
and program analyses (reaching definitions for stack variables and trace partitioning \citep{mauborgne2005trace}).

\begin{figure*}
  \hrule
    \vrule
  \begin{minipage}{0.33\linewidth}
\begin{code}[frame=none]
#include <stdlib.h>

struct LL
{
    struct LL * next;
    int handle;
};
    
int close_last(struct LL * list)
{
    while (list->next != NULL)
    {
        list = list->next;
    }
    return close(list->handle);
}
\end{code}
\end{minipage}
\vrule
\begin{minipage}{0.34\linewidth}
\begin{code}[frame=none]
close_last:
   push    ebp
   mov     ebp,esp
   sub     esp,8
   mov     edx,dword [ebp+arg_0]
   jmp     loc_8048402
loc_8048400:
   mov     edx,eax
loc_8048402:
   mov     eax,dword [edx]
   test    eax,eax
   jnz     loc_8048400
   mov     eax,dword [edx+4]
   mov     dword [ebp+arg_0],eax
   leave
   jmp     __thunk_.close
\end{code}
\end{minipage}
\vrule
\begin{minipage}{0.32\linewidth}
\hspace{0.1in}
\resizebox{0.9\linewidth}{!}{
\begin{minipage}{0.32\linewidth}
\begin{align*}
\scheme{F}{(\exists \tau . \mathcal{C})}{F} & \quad\text{where}\quad \mathcal{C} = \\
F.\mathsf{in_{stack0}} &\subtype \tau \\
\tau.\mathsf{load}.\sigma\mathsf{32@0} & \subtype \tau \\
\tau.\mathsf{load}.\sigma\mathsf{32@4} & \subtype \mathsf{int} \wedge ^\#\!\mathsf{FileDescriptor} \\
\mathsf{int} \vee ^\#\!\mathsf{SuccessZ} & \subtype F.\mathsf{out_{eax}}
\end{align*}
\end{minipage}
}
\vspace{0.1in}
\hrule
\hspace{0.09in}
\begin{minipage}{0.32\linewidth}
\begin{code}[frame=none]
typedef struct {
    Struct_0 * field_0;
    int // #FileDescriptor
      field_4;
} Struct_0;

int // #SuccessZ
close_last(const Struct_0 *);
\end{code}
\end{minipage}
\end{minipage}
\hspace{-5pt}
\vrule
\vspace{-4pt}
\caption{Example C code (compiled with {\tt gcc} 4.5.4 on Linux with flags {\tt -m32 -O2}), disassembly, type scheme inferred from the machine code, and reconstructed C type. The tags $^\#\!\mathsf{FileDescriptor}$ and $^\#\!\mathsf{SuccessZ}$ encode inferred higher-level purposes.}
\label{type-recovery-example}
\end{figure*}
 
\subsection{Polymorphic Functions}
\label{polyfun}
We discovered that, although not directly supported by the C type system, most programs define or make use of functions that are effectively polymorphic. The most well-known example is {\verb|malloc|}: the return value is expected to be immediately cast to some other
type {\verb|T*|}.  Each call to {\verb|malloc|} may be thought of as returning
some pointer of a different type.  The type of {\verb|malloc|} is effectively
{\em not} $\texttt{size\_t} \to \texttt{void*}$, but rather $\forall \tau. \texttt{size\_t} \to \tau\texttt{*}$.

The problem of a polymorphic {\verb|malloc|} could be mitigated by treating each call site $p$ as a call to a distinct function {\verb|malloc|}$_p$,
 each of which may have a distinct return type {\tt T$_p$*}.  Unfortunately, it is not sufficient to treat a handful of special functions like {\verb|malloc|} this way: it is common to see binaries that
 use user-defined allocators and wrappers to {\verb|malloc|}.  All of these functions would also need to be accurately identified and duplicated for each callsite.
 
A similar problem exists for functions like {\verb|free|}, which is polymorphic in its lone parameter.  Even more complex are functions like {\verb|memcpy|}, which is polymorphic in its first two parameters and its return type, though the three types are not independent of each other.  Furthermore,
the polymorphic type signatures
\begin{align*}
&\texttt{malloc} : \forall \tau . \texttt{size\_t} \to \tau\texttt{*} \\
&\texttt{free} : \forall \tau . \tau\texttt{*} \to \texttt{void} \\
&\texttt{memcpy} : \forall \alpha, \beta . (\beta \subtype \alpha) \Rightarrow (\alpha\texttt{*} \times \beta\texttt{*} \times \texttt{size\_t}) \to \alpha\texttt{*}
\end{align*}
are all strictly more informative to the reverse engineer than the standard C signatures.  How else
could one know that the {\verb|void*|} returned by {\verb|malloc|} is not meant to be an
opaque handle, but rather should be cast to some other pointer type?

In compiled C++ binaries, polymorphic functions are even more common. For example, a class member function must potentially accept
both {\verb|base_t*|} and {\verb|derived_t*|} as types for {\verb|this|}.

\citet{foster2006flow} noted that using bounded polymorphic
type schemes for {\verb|libc|} functions increased the precision of type-qualifier inference, at the level of source code.
To advance the state of the art in machine-code type recovery, we believe it is important to also
embrace polymorphic functions as a natural and common feature of machine code.
Significant improvements to static type reconstruction---even for monomorphic types---will
require the capability to infer polymorphic types of some nontrivial complexity.

\subsection{Recursive Types}

The relevance of recursive types for decompilation was recently discussed by \citet{schwartz}, where
lack of a recursive type system for machine code was cited as an important source of imprecision.
Since recursive data structures are relatively common, it is desirable that a type-inference scheme for machine code be able to represent and infer recursive types natively.

\subsection{Offset and Reinterpreted Pointers}
\label{ptr-to-member}

Unlike in source code, there is no syntactic distinction in machine code between a pointer-to-{\verb|struct|} and
a pointer-to-first-member-of-{\verb|struct|}.
For
example, if $X$ has type 
{\verb|struct|} {\verb|{ char*, FILE*, size_t }*|} on a 32-bit platform, then it should be possible to infer that $X+4$ can be safely passed to {\verb|fclose|}; conversely, if $X+4$ is passed
to {\verb|fclose|} we may need to 
infer that $X$ points to a structure that, at offset 4, contains a {\verb|FILE*|}.
This affects the typing of local structures, as well: a structure on the
stack may be manipulated using a pointer to its starting address or by manipulating the members 
directly, e.g., through the frame pointer.

These idioms, along with casts from {\verb|derived*|} to {\verb|base*|}, fall under the
general class of {\em physical} \citep{siff1999coping} or {\em non-structural}
\citep{palsberg1997type} subtyping. In Retypd, we model these forms of
subtyping using type scheme specialization (\autoref{sketch-sec}).
Additional hints about the extent
of local variables are found using data-delineation analysis \citep{param-offset}.

\subsection{Disassembly Failures}
\label{dis-fail}
The problem of producing correct disassembly for stripped binaries is equivalent to the halting problem. As a result,
we can never assume that our reconstructed program representation will be perfectly correct.
Even sound analyses built on top of an unsound program representation may exhibit inconsistencies and quirks.

Thus, we must be careful that incorrect disassembly or analysis results from one part of the binary
will not influence the correct type results we may have gathered for the rest of the binary. 
Type systems that model value assignments 
as type unifications are vulnerable to over-unification issues caused by bad IR.  Since unification is
non-local, bad constraints in one part of the binary can degrade {\em all} type results.

Another instance of this problem arises from the use of register parameters.
Although the x86 {\tt cdecl} calling convention uses the stack for parameter passing, most optimized
binaries will include many functions that pass parameters in registers for speed.  Often, these functions do not conform to any standard calling convention. Although we work hard to ensure that only
true register parameters are reported, conservativeness demands the occasional false positive.

Type-reconstruction methods that are based on unification are generally sensitive to precision
loss due to false-positive register parameters.
A common case is the ``{\verb|push ecx|}'' idiom that reserves space for a single
local variable in the stack frame of a function $f$.  If {\verb|ecx|} is incorrectly viewed as a
register parameter of $f$ in a unification-based
scheme, whatever type variables are bound to {\verb|ecx|} at each
callsite to $f$ will be mistakenly unified.  In our early experiments, we found these overunifications
to be a persistent and hard-to-diagnose source of imprecision.

In our early unification-based experiments, mitigation heuristics against overunification quickly
ballooned into a disproportionately large and unprincipled component of type analysis.
We designed Retypd's subtype-based constraint system to avoid the need for such ad-hoc
prophylactics against overunification.

\subsection{Cross-casting and Bit Twiddling}
\label{twiddle}
Even at the level of source code, there are already many type-unsafe idioms in
common use.  Most of these idioms operate by directly manipulating the bit representation of
a value, either to encode additional information or to perform computations that are not
possible using the type's usual interface.  Some common examples include

\begin{itemize}[noitemsep]
\item hashing values by treating them as untyped bit blocks \citep{tr1},
\item stealing unused bits of a pointer for tag information, such as whether a thunk has been
evaluated \citep{spj},
\item reducing the storage requirements of a doubly-linked list by \textsc{xor}-combining the
{\verb|next|} and {\verb|prev|} pointers, and
\item directly manipulating the bit representation of another type, as in the {\verb|quake3|}
inverse square root trick \citep{inv-sqrt}.
\end{itemize}

Because of these type-unsafe idioms, it is important that a type-inference scheme continues to produce useful results even in the presence of apparently contradictory constraints.  We handle this situation in three ways:
\begin{enumerate}[noitemsep]
\item separating the phases of constraint entailment, solving, and consistency checking,
\item modeling types with sketches (\autoref{sketch-sec}) that carry more information than C types, and
\item using unions to combine types with otherwise incompatible capabilities
(e.g., $\tau$ is both {\verb|int|}-like and pointer-like).
\end{enumerate}

\subsection{Incomplete Points-to Information}
 Degradation of points-to accuracy on large programs has been identified
as a source of type-precision loss in other systems \citep{2ndwrite}.
Our algorithm can provide high-quality types even in the absence of points-to information.  Precision can be further improved by increasing points-to knowledge via machine-code analyses such as VSA \citep{Balakrishnan2004}, but good results are already attained with no points-to analysis beyond the simpler problem of tracking the
stack pointer.

\subsection{Ad-hoc Subtyping}
\label{ad-hoc}
Programs may define an ad-hoc type hierarchy via typedefs.  This idiom appears in the Windows API, where
a variety of handle types are all defined as typedefs of {\tt void*}.  Some of the handle types are to be
used as subtypes of other handles; for example, a GDI handle ({\tt HGDI}) is a generic handle used to
represent any one of the more specific {\tt HBRUSH}, {\tt HPEN}, {\it etc}.
In other cases, a typedef may indicate a {\em supertype}, as in {\verb|LPARAM|} or {\verb|DWORD|};
 although these are typedefs of {\verb|int|}, they have the intended semantics of a generic 32-bit type,
 which in different contexts may be used as a pointer, an integer, a flag set, and so on.

To accurately track ad-hoc hierarchies requires a type system based around subtyping rather than
unification. Models for common API type hierarchies are useful; still better is the ability for the end
user to define or adjust the initial type hierarchy at run time.  We support this feature by parameterizing
the main type representation by an uninterpreted lattice $\Lambda$, as described in \autoref{sketch-sec}.

\section{The Type System}

\begin{table}
\centering
\caption{Example field labels (type capabilities) in $\Sigma$.}
\label{label-table}
\vspace{0.5ex}
\begin{tabular}{@{}rcl@{}}
\toprule
Label & Variance & Capability \\
\midrule
$\mathsf{.in}_L$ & $\ominus$ & Function with input in location $L$. \\
$\mathsf{.out}_L$ & $\oplus$ & Function with output in location $L$. \\
$\mathsf{.load}$ & $\oplus$ & Readable pointer.\\
$\mathsf{.store}$ & $\ominus$ & Writable pointer.\\
$.\sigma\mathsf{N@k}$ & $\oplus$ & Has $N$-bit field at offset $k$. \\
\bottomrule
\end{tabular}
\end{table}

\begin{figure*}
\hrule
\vspace{0.1in}
\begin{minipage}{0.43\linewidth}
\resizebox{0.9\textwidth}{!}{
\begin{minipage}{\linewidth}
\centering
{\it Derived Type Variable Formation}
\[
\inferrule
  {\alpha \subtype \beta}
  {\term{\alpha}}
  \; {\textsc {(T-Left)}}
\hspace{0.25in}
\inferrule
  {\alpha \subtype \beta, \quad \term{\alpha.\ell}}
  {\term{\beta.\ell}}
  \; {\textsc {(T-InheritL)}}
\]\[
\inferrule
  {\alpha \subtype \beta}
  {\term{\beta}}
  \; {\textsc {(T-Right)}}
\hspace{0.25in}
\inferrule
  {\alpha \subtype \beta, \quad \term{\beta.\ell}}
  {\term{\alpha.\ell}}
  \; {\textsc {(T-InheritR)}}
\]\[
\inferrule
  {\term{\alpha.\ell}}
  {\term{\alpha}}
  \; {\textsc {(T-Prefix)}}
\]
\end{minipage}}
\end{minipage}
\vrule
\begin{minipage}{0.56\linewidth}
\resizebox{0.9\textwidth}{!}{
\begin{minipage}{\linewidth}
\centering
{\it Subtyping}
\[
\hspace{0.5in}
\inferrule
  {\term{\alpha}}
  {\alpha \subtype \alpha}
  \; {\textsc {(S-Refl)}}
\hspace{0.2in}
\inferrule
  {\alpha \subtype \beta, \quad \term{\beta.\ell}, \quad \langle \ell \rangle = \oplus}
  {\alpha.\ell \subtype \beta.\ell}
  \; {\textsc {(S-Field$_\oplus$)}}
\]\[
\inferrule
  {\alpha \subtype \beta ,\quad \beta \subtype \gamma}
  {\alpha \subtype \gamma}
  \; {\textsc {(S-Trans)}}
\hspace{0.2in}
\inferrule
  {\alpha \subtype \beta, \quad \term{\beta.\ell}, \quad \langle \ell \rangle = \ominus}
  {\beta.\ell \subtype \alpha.\ell}
  \; {\textsc {(S-Field$_\ominus$)}}
\]\[
\inferrule
  {\term{\alpha.\mathsf{load}}, \quad \term{\alpha.\mathsf{store}}}
  {\alpha.\mathsf{store} \subtype \alpha.\mathsf{load}}
  \; {\textsc {(S-Pointer)}}
\]
\end{minipage}}
\end{minipage}

\caption{Deduction rules for the type system.  $\alpha, \beta, \gamma$ represent derived
type variables; $\ell$ represents a label in $\Sigma$.}
\label{deduction-rules}
\end{figure*}
 
The type system used by Retypd is based around the inference of
{\em recursively constrained type schemes} (\autoref{section-syntax}). Solutions
to constraint sets are modeled by {\em sketches} (\autoref{sketch-sec});
the sketch associated to a value
consists of a record of capabilities which that value holds, such
as whether it can be stored to, called, or accessed at a certain offset.
Sketches also include markings drawn from a customizable lattice
$(\Lambda, \vee, \wedge, <:)$,
used to propagate high-level information such as typedef
names and domain-specific purposes during type inference.

Retypd also supports recursively constrained type schemes that abstract
over the set of types subject to a constraint set $\mathcal{C}$.
The language of type constraints used by Retypd is weak
enough that for any constraint set $\mathcal{C}$, satisfiability of 
$\mathcal{C}$ can be reduced (in cubic time) to checking a set of
scalar constraints $\kappa_1 <: \kappa_2$,
 where $\kappa_i$ are constants belonging to $\Lambda$.

Thanks to the reduction of constraint satisfiability to scalar constraint checking,
we can omit expensive satisfiability checks during type inference.
Instead, we delay the  check
until the final stage when internal types are converted to C types for display,
providing a natural place to instantiate union types that resolve any inconsistencies.
Since compiler optimizations and type-unsafe idioms in the original source
frequently lead to program fragments with unsatisfiable type constraints
(\autoref{dis-fail}, \autoref{twiddle}),
this trait is  particularly desirable.

\subsection {Syntax: the Constraint Type System}
\label{section-syntax}

Throughout this section, we fix a set $\mathcal{V}$ of {\em type variables}, an alphabet
$\Sigma$ of {\em field labels}, and a function $\langle \cdot \rangle : \Sigma \to \{ \oplus, \ominus \}$
denoting the {\em variance} (\autoref{variance-def}) of each label.  We do not require the set $\Sigma$ to be finite.
Retypd makes use of a large set of labels; for simplicity, we will focus on those in
\autoref{label-table}.

Within $\mathcal{V}$, we assume there is a distinguished set of {\em type constants}.  These type constants
are symbolic representations $\overline{\kappa}$ of elements $\kappa$ belonging to some lattice, but
are otherwise uninterpreted.  It is usually sufficient to think of the type constants as type names or
semantic tags.

\begin{definition}
  A {\em derived type variable} is an expression of the form $\alpha w$ with
  $\alpha \in \mathcal{V}$ and $w \in \Sigma^*$.
\end{definition}

\begin{definition}
\label{variance-def}
  The variance of a label $\ell$ encodes the subtype relationship between $\alpha.\ell$
  and $\beta.\ell$ when $\alpha$ is a subtype of $\beta$, formalized in rules
  $\textsc{S-Field}_\oplus$ and $\textsc{S-Field}_\ominus$ of \autoref{deduction-rules}.
  The variance function $\langle \cdot \rangle$ can be extended to $\Sigma^*$ by defining
$\langle \varepsilon \rangle = \oplus$ and
$\langle xw \rangle = \langle x \rangle \cdot \langle w \rangle$,
where $\{\oplus, \ominus\}$ is the sign monoid with
  $\oplus \cdot \oplus = \ominus \cdot \ominus = \oplus$ and 
  $\oplus \cdot \ominus = \ominus \cdot \oplus = \ominus$.
  A word $w \in \Sigma^*$
    is
  called {\em covariant} if $\langle w \rangle = \oplus$, or {\em contravariant} if
  $\langle w \rangle = \ominus$.
\end{definition}

\begin{definition}
\label{constraint-def}
  Let $\mathcal{V} = \{\alpha_i\}$ be a set of base type variables.
  A {\em constraint} is an expression of the form $\term{X}$ (``existence of
  the derived type variable $X$'') or $X \subtype Y$ (``$X$ is a subtype of $Y$''), where
  $X$ and $Y$ are derived type variables.
A {\em constraint set over $\mathcal{V}$} is a finite
  collection $\mathcal{C}$ of constraints, where the type variables in each constraint are
  either type constants or members of $\mathcal{V}$. We will say that $\mathcal{C}$ {\em entails} $c$,
  denoted $\mathcal{C} \proves c$, if $c$ can be derived from the constraints in $\mathcal{C}$
  using the deduction rules of \autoref{deduction-rules}. We also allow projections:
  given a constraint set $\mathcal{C}$ with free variable $\tau$, the projection
  $\exists \tau . \mathcal{C}$ binds $\tau$ as an ``internal'' variable in the constraint
  set. See $\tau$ in \autoref{type-recovery-example} for an example or, for a more in-depth
  treatment of constraint projection, see \citet{su2002first}.
\end{definition}

The field labels used to form derived type variables are meant to represent {\em capabilities} of
a type.  For example, the constraint $\term{\alpha.\mathsf{load}}$ means
$\alpha$ is a readable pointer, and the derived type variable $\alpha.\mathsf{load}$
represents the type of the memory region obtained by loading from $\alpha$.

Let us briefly see how operations in the original program translate to type constraints,
using C-like pseudocode for clarity.
The full conversion from disassembly to type constraints is described in 
\aautoref{sec-interp}{Appendix A of \citep{arxiv}}.

\smallskip
\noindent
{\bf Value copies:}
When a value is moved between program variables in an assignment like {\verb|x := y|}, we
make the conservative assumption that the type of {\verb|x|} may be upcast to a supertype of
{\verb|y|}.
We will generate a constraint of the form $Y \subtype X$.

\smallskip
\noindent
{\bf Loads and stores:}
Suppose that {\verb|p|} is a pointer to a 32-bit type,
and a value is loaded into {\verb|x|} by the assignment {\verb|x := *p|}.
Then we will generate a constraint of the form $P.\mathsf{load}.\sigma\mathsf{32@0} \subtype X$.
Similarly, a store {\verb|*q := y|} results in the constraint
$Y \subtype Q.\mathsf{store}.\sigma\mathsf{32@0}$.

In some of the pointer-based examples in this paper we omit the final $.\sigma\mathsf{N@k}$ access after a
$.\mathsf{load}$ or $.\mathsf{store}$ to simplify the presentation.

\smallskip
\noindent
{\bf Function calls:}
Suppose the function $f$ is invoked by
{\verb|y := f(x)|}. We will generate the constraints $X \subtype F.\mathsf{in}$ and
$F.\mathsf{out} \subtype Y$, reflecting the flow of actuals to and from formals. Note that
if we define $A.\mathsf{in} = X$ and $A.\mathsf{out} = Y$ then the two constraints are equivalent
to $F \subtype A$ by the rules of \autoref{deduction-rules}. This encodes the fact that the
called function's type must be at least as specific as the type used at the callsite.

One of the primary goals of our type-inference engine is to associate to each procedure a
most-general type scheme.
\begin{definition}
A {\em type scheme} is an expression of the form $\scheme{\overline{\alpha}}
{\mathcal{C}}{\alpha_1}$
where $\forall \overline{\alpha} = \forall \alpha_1 \dots \forall \alpha_n$
is quantification over a set
of type variables, and $\mathcal{C}$ is a constraint set over $\{\alpha_i\}_{i=1..n}$.
\end{definition}
Type schemes provide a way of encoding the pre- and post-conditions that a function places on
the types in its calling context.  Without the constraint sets, we would only be able to 
represent conditions of the form ``the input must be a subtype of $X$'' and ``the output must be
a supertype of $Y$''.  The constraint set $\mathcal{C}$ can be used to encode more
interesting type relations between inputs and outputs, as in the case of
{\verb|memcpy|} (\autoref{polyfun}).  For example, a function that 
returns the second 4-byte element from a {\verb|struct*|} may have the type scheme
$\scheme{\tau}{(\tau.\mathsf{in}.\mathsf{load}.\sigma\mathsf{32@4} \subtype \tau.\mathsf{out})}{\tau}$.

\subsection{Deduction Rules}
\label{deduction-rule-sec}

The deduction rules for our type system appear in \autoref{deduction-rules}.
Most of the rules are self-evident under the interpretation in
\autoref{constraint-def}, but a few require some additional motivation.

\smallskip
\noindent
{\bf $\textsc{S-Field}_\oplus$ and $\textsc{S-Field}_\ominus$:}
These rules ensure that field labels act as co- or contra-variant type operators,
generating subtype relations between derived type variables from subtype relations between
the original variables.

\smallskip
\noindent
{\bf $\textsc{T-InheritL}$ and $\textsc{T-InheritR}$:}
The rule $\textsc{T-InheritL}$
should be uncontroversial, since a subtype should have all capabilities of its supertype.
The rule $\textsc{T-InheritR}$ is more unusual since it moves capabilities in the other
direction; taken together, these rules require that
two types in a subtype relation must have exactly the same set of capabilities.
This is a form of structural typing, ensuring that comparable types have the
same shape.

Structural typing appears to be at odds with the need to cast more capable objects to
less capable ones, as described in \autoref{ptr-to-member}.  Indeed, $\textsc{T-InheritR}$
eliminates the possibility of
forgetting capabilities during value assignments.  But we still maintain this capability at procedure
invocations due to our use of  polymorphic type schemes.  An
explanation of how type-scheme instantiation enables us to forget
fields of an object appears in \autoref{narrowing}, with more details in \aautoref{reverse-dns}{\S E.1.2 of \citep{arxiv}}.

These rules ensure that Retypd can perform ``iterative variable recovery''; lack of
iterative variable recovery was cited by the creators of the Phoenix decompiler \citep{schwartz}
as a common cause of incorrect decompilation when using TIE \citep{tie} for type recovery.

\smallskip
\noindent
{\bf $\textsc{S-Pointer}$:} This rule is a consistency condition ensuring that
the type that can be loaded from a pointer is a supertype of the type that can be stored to
a pointer.  Without this rule, pointers would provide a channel for subverting the
type system.  An example of how this rule is used in practice
appears in \autoref{model-ptrs}.
\smallskip

The deduction rules of \autoref{deduction-rules} are simple enough that each proof may
be reduced to a normal form (see \aautoref{normal-form-thm}{Theorem B.1 in \citep{arxiv}}).  An encoding of the normal forms as
transition sequences in a modified pushdown system is used to provide a compact representation
of the entailment closure $\overline{\mathcal{C}} = \{ c ~|~ \mathcal{C} \proves c\}$.
The pushdown system modeling $\overline{\mathcal{C}}$ is queried and manipulated to 
provide most of the interesting type-inference functionality.  An outline of this functionality
appears in \autoref{pds-section}.

\subsection{Modeling Pointers}
\label{model-ptrs}
To model pointers soundly in the presence of subtyping, we found that our initial na\"ive
approach suffered from unexpected difficulties when combined with subtyping.  Following the
C type system, it seemed natural to model pointers by introducing an injective unary
type constructor $\mathsf{Ptr}$, so that $\mathsf{Ptr}(T)$ is the type of pointers
to $T$. In a unification-based type system, this approach works as expected.

In the presence of subtyping, a new issue arises. Consider the two programs in \autoref{aliased-ptrs}.
Since the type variables $P$ and $Q$ associated to {\verb|p|}, {\verb|q|} can be seen to be pointers,
we can begin by writing $P = \mathsf{Ptr}(\alpha), Q = \mathsf{Ptr}(\beta)$.
The first program will generate the constraint set
$\mathcal{C}_1 = \{ \mathsf{Ptr}(\beta) \subtype \mathsf{Ptr}(\alpha),$ $ X \subtype \alpha,$ 
$ \beta \subtype Y\}$
while the second generates
$\mathcal{C}_2 = \{ \mathsf{Ptr}(\beta) \subtype \mathsf{Ptr}(\alpha),$ $ X \subtype \beta,$ 
$ \alpha \subtype Y\}$.
Since each program has the effect of copying the value in $x$ to $y$, both constraint sets should
satisfy $\mathcal{C}_i \proves X \subtype Y$.  To do this, the pointer subtype constraint
must entail some constraint on $\alpha$ and $\beta$, but which one?

If we assume that $\mathsf{Ptr}$ is covariant, then $\mathsf{Ptr}(\beta) \subtype \mathsf{Ptr}(\alpha)$
entails $\beta \subtype \alpha$ and so $\mathcal{C}_2 \proves X \subtype Y$,
but $\mathcal{C}_1 \notproves X \subtype Y$.  On the other hand, if we make $\mathsf{Ptr}$ contravariant
then $\mathcal{C}_1 \proves X \subtype Y$ but $\mathcal{C}_2 \notproves X \subtype Y$.

It seems that our only recourse is to make subtyping degenerate to type equality under $\mathsf{Ptr}$:
we are forced to declare that $\mathsf{Ptr}(\beta) \subtype \mathsf{Ptr}(\alpha) \proves \alpha = \beta$,
which of course means that $\mathsf{Ptr}(\beta) = \mathsf{Ptr}(\alpha)$ already.
This is a catastrophe for subtyping as used in machine code, since many natural subtype relations are
mediated through pointers.  For example, the unary $\mathsf{Ptr}$ constructor cannot handle the
simplest kind of C++ class subtyping, where a derived class physically extends a base class by appending
new member variables.

The root cause of the difficulty seems to be in conflating two capabilities that (most) pointers have:
the ability to be written through and the ability to be read through.  In Retypd, these two capabilities
are modeled using different field labels $\mathsf{.store}$ and $\mathsf{.load}$.  The $\mathsf{.store}$
label is contravariant, while the $\mathsf{.load}$ label is covariant.

\begin{figure}
  \hrule
  \vrule
\begin{minipage}{0.48\linewidth}
\begin{code}[frame=none]
    f() {
       p =  q;
      *p =  x;
       y = *q;
    }
\end{code}
\end{minipage}
\vrule
\begin{minipage}{0.48\linewidth}
\begin{code}[frame=none]
    g() {
       p =  q;
      *q =  x;
       y = *p;
    }
\end{code}
\end{minipage}
\hspace{2pt}
\vrule
\vspace{-4pt}
\caption{Two programs, each mediating a copy from {\small{\tt x}} to {\small{\tt y}} through a pair of aliased pointers.}
\label{aliased-ptrs}
\end{figure}

To see how the separation of pointer capabilities avoids the loss of precision suffered by $\mathsf{Ptr}$,
we revisit the two example programs.  The first generates the constraint set
\[\mathcal{C}'_1 = \left\{ Q \subtype P,\quad X \subtype P\mathsf{.store},\quad Q\mathsf{.load} \subtype Y \right\}\]
By rule $\textsc{T-InheritR}$ we may conclude that $Q$ also has a field of type $\mathsf{.store}$.  By
$\textsc{S-Pointer}$ we can infer that $Q\mathsf{.store} \subtype Q\mathsf{.load}$.  Finally, since $\mathsf{.store}$
is contravariant and $Q \subtype P$, $\textsc{S-Field}_\ominus$ says we also have
$P\mathsf{.store} \subtype Q\mathsf{.store}$.  Putting these parts together gives the subtype chain
\[X \subtype P\mathsf{.store} \subtype Q\mathsf{.store} \subtype Q\mathsf{.load} \subtype Y\]

\noindent
The second program generates the constraint set
\[\mathcal{C}'_2 = \left\{ Q \subtype P,\quad X \subtype Q\mathsf{.store},\quad
P\mathsf{.load} \subtype Y \right\}\]
Since $Q \subtype P$ and $P$ has a field $\mathsf{.load}$, we conclude that $Q$ has a $\mathsf{.load}$ field
as well.  Next, $\textsc{S-Pointer}$ requires that $Q\mathsf{.store} \subtype Q\mathsf{.load}$.
Since $\mathsf{.load}$ is covariant, $Q \subtype P$ implies that $Q\mathsf{.load} \subtype P\mathsf{.load}$.
This gives the subtype chain
\[X \subtype Q\mathsf{.store} \subtype Q\mathsf{.load} \subtype P\mathsf{.load} \subtype Y\]

By splitting out the read- and write-capabilities of a pointer, we can achieve a sound account of 
pointer subtyping that does not degenerate to type equality.  Note the importance
of the consistency condition $\textsc{S-Pointer}$: this rule ensures that writing through a pointer
and reading the result cannot subvert the type system.

The need for separate handling of read- and write-capabilities in a mutable reference has been rediscovered
multiple times.  A well-known instance is the covariance of the array type constructor in Java
and C\#, which can
cause runtime type errors if the array is mutated; in these languages, the read capabilities are soundly modeled only
by sacrificing soundness for the write capabilities.

\subsection{Non-structural Subtyping and $\textsc{T-InheritR}$}
\label{narrowing}

It was noted in \autoref{deduction-rule-sec} that the rule $\textsc{T-InheritR}$ leads to a
system with a form of structural typing: any two types in a subtype relation must have the same
capabilities. Superficially, this seems
problematic for modeling typecasts that forget about fields, such as a cast from
{\verb|derived*|} to {\verb|base*|} when {\verb|derived*|} has additional
fields (\autoref{ptr-to-member}).

The missing piece that allows us to effectively forget capabilities is instantiation of callee type schemes at a callsite.  To demonstrate how polymorphism enables forgetfulness, consider the 
example type scheme $\scheme{F}{(\exists \tau . \mathcal{C})}{F}$ from \autoref{type-recovery-example}.
The function {\verb|close_last|} can be invoked by providing any actual-in type $\alpha$,
such that $\alpha \subtype F.\mathsf{in}_{\mathsf{stack}0}$; in particular, $\alpha$ can have
{\em more} fields than those required by $\mathcal{C}$.  We simply select a
more capable type for the existentially-quantified type variable $\tau$ in $\mathcal{C}$.
In effect, we have
used {\em specialization} of polymorphic types to model {\em non-structural}
subtyping idioms, while {\em subtyping} is used only to model {\em structural} subtyping idioms.  This restricts our introduction of non-structural subtypes to points where a
type scheme is instantiated, such as at a call site.

\subsection {Semantics: the Poset of Sketches}
\label{sketch-sec}
\begin{figure}
\[\begin{tikzpicture}[scale=0.7]
\node[shape=circle,draw] (root)   at (-1,1.25) {\scriptsize$\top$};
\node[shape=circle,draw] (in)     at (1,2) {\scriptsize$\top$};
\node[shape=circle,draw] (out)    at (1,0.5) {\scriptsize$\alpha$};
\node[shape=circle,draw] (loaded) at (3,2) {\scriptsize$\top$};
\node[shape=circle,draw] (atZero) at (5,3) {\scriptsize$\top$};
\node[shape=circle,draw] (atFour) at (5,1) {\scriptsize$\beta$};
\node (recurse) at (6.5,3) {};
\draw (root) -- (in) -- (loaded) -- (atZero) -- (recurse);
\draw                   (loaded) -- (atFour);
\draw (root) -- (out);
\draw[dashed] (2.4,2) -- (3.5,3.8) -- (8,3.8) -- (8,0.2) -- (3.5,0.2) -- cycle;
\draw (6.3,3) -- (6.5,2.3) -- (7.8,2.3) -- (7.8,3.7) -- (6.5,3.7) -- cycle;
\node at (7.5,0.6) {{$A$}};
\node at (7,3) {{$A$}};
\node[rotate=25] at (-0.1,2) {\scriptsize${\mathsf{in}_{\mathsf{stack}0}}$};
\node[rotate=-25] at (-0.1,0.5) {\scriptsize{$\mathsf{out}_{\mathsf{eax}}$}};
\node at (1.9, 2.2) {\scriptsize{$\mathsf{load}$}};
\node[rotate=30]  at (3.9, 2.8) {\scriptsize{$\sigma\mathsf{32@0}$}};
\node[rotate=-30] at (3.9, 1.2) {\scriptsize{$\sigma\mathsf{32@4}$}};
\node at (5.9,3.2) {\scriptsize{$\mathsf{load}$}};
\node at (0,3.5) {\scriptsize{$\alpha = \mathsf{int} \vee ^\#\!\mathsf{SuccessZ}$}};
\node at (0.35,3) {\scriptsize{$\beta = \mathsf{int} \wedge ^\#\!\mathsf{FileDescriptor}$}};
\end{tikzpicture}\]
\caption{A sketch instantiating the type scheme in \autoref{type-recovery-example}.}
\label{sketch-sketch}
\end{figure}
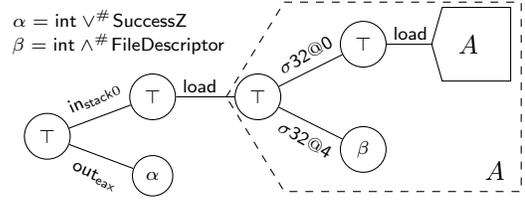

The simple type system defined by the deduction rules of \autoref{deduction-rules}
defines the {\em syntax} of legal derivations in our type system.  The constraint
solver of \autoref{pds-section} is designed to find a simple representation
for all conclusions that can be derived from a set of type constraints.  Yet there
is no notion of what a type {\em is} inherent to the deduction rules of
\autoref{deduction-rules}.  We have defined the rules of the game, but not
the equipment with which it should be played.

We found that introducing C-like entities at the level of constraints or types resulted in
too much loss of precision when working with the challenging examples described in
\autoref{challenge-sec}. Consequently we developed the notion of a {\em sketch},
a kind of regular tree labeled with elements of an auxiliary lattice $\Lambda$.
Sketches are related to the recursive types studied by \citet{cardelli} and \citet{kozen},
but do not depend on {\it a priori} knowledge of the ranked alphabet of type constructors.

\begin{definition}
\label{sketch-def}
A {\em sketch} is a (possibly infinite) tree $T$ with edges labeled by elements of $\Sigma$ and nodes marked with elements of a lattice $\Lambda$, such that $T$ only has finitely many subtrees up to
labeled isomorphism. By collapsing isomorphic subtrees, we can represent sketches as deterministic
finite state automata with
each state labeled by an element of $\Lambda$. The set of sketches admits a lattice structure, with
operations described by \aautoref{sketch-lattice}{Figure 18 in \citep{arxiv}}.

The lattice of sketches serves as the model in which we interpret type constraints. The
interpretation of the constraint $\term{\alpha.u}$ is ``the sketch $S_\alpha$ admits a path
from the root with label sequence $u$'', and $\alpha.u \subtype \beta.v$ is interpreted as
``the sketch obtained from $S_\alpha$ by traversing the label sequence $u$ is a subsketch
(in the lattice order) of the sketch obtained from $S_\beta$ by traversing the sequence $v$.''
\end{definition}

The main utility of sketches is that they are nearly a free tree model
\citep{pottier2005essence} of the constraint language.  
Any constraint set $\mathcal{C}$ is satisfiable over the lattice of sketches, as long as
$\mathcal{C}$ cannot prove an impossible subtype relation in the auxiliary lattice $\Lambda$.
In particular, we can always solve the fragment of $\mathcal{C}$ that does not reference
constants in $\Lambda$.  Stated operationally, we can always recover the tree structure of
sketches that potentially solve $\mathcal{C}$. This observation is formalized by the
following theorem:
\begin{theorem}
\label{recognize-var}
Suppose that $\mathcal{C}$ is a constraint set over the variables $\{\tau_i\}_{i \in I}$.
Then there exist sketches $\{S_i\}_{i \in I}$, such that $w \in S_i$ if and only if
$\mathcal{C} \proves \term{\tau_i.w}$.
\end{theorem}
\begin{proof}
The idea is to symmetrize $\sqsubseteq$ using an algorithm that is similar in
spirit to Steensgaard's
method of almost-linear-time pointer analysis \citep{steensgaard}.
Begin by forming a graph with one
node $n(\alpha)$ for each
derived type variable appearing in $\mathcal{C}$, along with each of its prefixes.  Add a labeled
edge $n(\alpha) \stackrel{\ell}{\to} n(\alpha.\ell)$  for each derived type variable $\alpha.\ell$ to
form a graph $G$.  Now quotient $G$ by the equivalence relation $\sim$ defined by $n(\alpha) \sim
n(\beta)$ if $\alpha \subtype \beta \in \mathcal{C}$, and $n(\alpha') \sim n(\beta')$ whenever there are edges
$n(\alpha) \stackrel{\ell}{\to} n(\alpha')$ and $n(\beta) \stackrel{\ell'}{\to} n(\beta')$ in $G$
with $n(\alpha) \sim n(\beta)$ where either $\ell = \ell'$ or $\ell = \mathsf{.load}$ and $\ell' = 
\mathsf{.store}$.

By construction, there exists a path through $G/\!\!\sim$ with label sequence $u$
starting at the equivalence class of $\tau_i$ if and only if
$\mathcal{C} \proves \term{\tau_i.u}$; the (regular) set of all
such paths yields the tree structure of $S_i$.
\end{proof}

Working out the lattice elements that should label $S_i$ is a trickier problem; the basic idea is to use
the same automaton $Q$ constructed during constraint simplification (\aautoref{mod-sat}{Theorem 5.1 in \citep{arxiv}}) to answer queries
about which type constants are upper and lower bounds on a given derived type variable.
The full algorithm is listed in \aautoref{shadowing}{\S D.4 of \citep{arxiv}}.

\smallskip
In Retypd, we use a large auxiliary lattice $\Lambda$ containing hundreds of elements that includes
a collection of standard C type names, common typedefs for popular APIs, and user-specified
semantic classes such as $^\#\!\mathsf{FileDescriptor}$ in \autoref{type-recovery-example}.
This lattice helps model ad-hoc subtyping
and preserve high-level semantic type names, as discussed in \autoref{ad-hoc}.

\cbstart
\smallskip
\noindent
{\bf Note.}
Sketches are just one of many possible models for the deduction rules that could be proposed.
A general approach is to fix a poset $(\mathcal{T}, <:)$ of types, interpret
$\subtype$ as $<:$, and interpret co- and contra-variant field labels as
monotone (resp. antimonotone) functions $\mathcal{T} \to \mathcal{T}$.

The separation of syntax from semantics allows for a simple way to parameterize the
type-inference engine by a model of types. By choosing a model $(\mathcal{T}, \equiv)$
with a symmetric relation $\equiv\; \subseteq \mathcal{T} \times \mathcal{T}$, a unification-based
type system similar to SecondWrite \citep{2ndwrite} is generated.
On the other hand, by forming a lattice of type intervals and interval inclusion,
we would obtain a type system similar to TIE \citep{tie} that outputs upper and
lower bounds on each type variable.
\cbend

\section {Analysis Framework}

\subsection{IR Reconstruction}
Retypd is built on top of GrammaTech's machine-code-analysis tool CodeSurfer for Binaries.
CodeSurfer carries out common program analyses on binaries for multiple CPU architectures,
including x86, x86-64, and ARM.  CodeSurfer is used to recover a high-level IR from the
raw machine code; type constraints are generated directly from this IR, and resolved types
are applied back to the IR and become visible to the GUI and later analysis phases.

CodeSurfer achieves platform independence through TSL \citep{lim12}, a language for defining a 
processor's concrete semantics in terms of concrete numeric types and mapping types that
model flag, register, and memory banks.  Interpreters for a given abstract domain are automatically
created from the concrete semantics simply
by specifying the abstract domain $\mathcal{A}$ and an interpretation of the concrete numeric
and mapping types.
Retypd uses CodeSurfer's recovered IR to determine the number and location of inputs and
outputs to each procedure, as well as the program's call graph and per-procedure control-flow
graphs.  An abstract interpreter then generates sets of type constraints from
the concrete TSL instruction semantics.  A detailed account of the abstract semantics
for constraint generation appears in \aautoref{sec-interp}{Appendix A of \citep{arxiv}}.

\subsection{Approach to Type Resolution}

After the initial IR is recovered, type inference proceeds in two stages: first,
type-constraint sets are generated in
a bottom-up fashion over the strongly-connected components of the callgraph.
Pre-computed type schemes for externally linked functions may be inserted at
this stage.
Each constraint set is simplified by eliminating type variables that do not
belong to the SCC's interface; the simplification algorithm is outlined in
\autoref{simplify}.
Once type schemes are available, the 
callgraph is traversed bottom-up, assigning sketches to type variables as
outlined in \autoref{sketch-sec}.
During this stage, type schemes are
specialized based on the calling contexts of each function.
\aautoref{the-algorithms}{Appendix F of \citep{arxiv}} lists the full algorithms for constraint simplification
(\aautoref{type-infer-alg}{\S F.1}) and solving (\aautoref{type-infer-alg2}{\S F.2}).

\subsection{Translation to \texttt{C} Types}
\label{policies}
The final phase of type resolution converts the inferred sketches to C types for presentation
to the user.
Since C types and sketches are not directly comparable, this resolution phase necessarily involves the application of heuristic conversion policies.
Restricting the heuristic policies to a single post-inference phase provides us with the
flexibility to generate high-quality, human-readable C types while maintaining soundness and generality
during type reconstruction.

\begin{example}
A simple example involves the generation of {\verb|const|} annotations on pointers. We decided
on a policy that only introduced {\verb|const|} annotations on function parameters, by
annotating the parameter at location $L$ when the constraint set $\mathcal{C}$
for procedure $p$ satisfies $\mathcal{C} \proves \term{p.\mathsf{in}_L.\mathsf{load}}$ and
$\mathcal{C} \notproves \term{p.\mathsf{in}_L.\mathsf{store}}$.
Retypd appears to be the first machine-code type-inference system to infer {\verb|const|}
annotations; a comparison of our recovered annotations to the original source code appears in
\autoref{const-recovery-stats}.
\label{const-recovery-info}
\end{example}

\begin{example}
A more complex policy is used to decide between union types and generic types when incompatible scalar constraints must be resolved. Retypd merges comparable scalar constraints to form antichains in $\Lambda$; the elements of these antichains are then used for the resulting C \verb|union| type.
\end{example}

\begin{example}
The initial type-simplification stage results in types that are as general as possible.
Often, this means that types are found to be more general than is strictly helpful to
a (human) observer. A policy is applied that specializes
type schemes to the most {\em specific} scheme that is compatible with all statically-discovered uses.
For example, a C++ object may include a getter function with a highly polymorphic
type scheme, since it could operate equally well on any structure with a field of the correct type at the
correct offset.  But we expect that in every calling context, the getter will be called on a specific 
object type (or perhaps its derived types).  We can specialize the getter's type by choosing the least polymorphic specialization that is compatible with the observed uses.  By specializing the function signature before presenting a final C type to the user, we trade some generality for types that are more likely to match the original source.
\end{example}

\section {The Simplification Algorithm}

In this section, we sketch an outline of the simplification algorithm at the core of
the constraint solver.  The complete algorithm appears in \aautoref{construction}{Appendix D of \citep{arxiv}}.
\label{simplify}

\subsection{Inferring a Type Scheme}
The goal of the simplification algorithm is to take an inferred type scheme
$\scheme{\overline{\alpha}}{C}{\tau}$ for a procedure and create a smaller constraint set
$\mathcal{C}'$, such that any constraint on $\tau$ implied by $\mathcal{C}$
is also implied by $\mathcal{C}'$.

Let $\mathcal{C}$ denote the constraint set generated by abstract interpretation
of the procedure being analyzed, and let $\overline{\alpha}$ be the set of
free type variables in $\mathcal{C}$.  We could already use
$\scheme{\overline{\alpha}}{\mathcal{C}}{\tau}$ as the
constraint set in the procedure's type scheme, since the
input and output types used in a valid invocation of {\verb|f|}
are tautologically those that satisfy $\mathcal{C}$.
Yet, as a practical matter, we cannot  use the constraint set directly, since this would result in
constraint sets with many useless free variables and a high growth rate over
nested procedures.

Instead, we seek to generate a {\em simplified constraint set} $\mathcal{C}'$,
such that if $c$ is an ``interesting'' constraint
and $\mathcal{C} \proves c$ then $\mathcal{C}' \proves c$ as well.
But what makes a constraint interesting?

\begin{definition}
For a type variable $\tau$, a constraint
is called {\em interesting} if it has
one of the following forms:
\begin{itemize}[noitemsep]
\item A capability constraint of the form $\term{\tau.u}$
\item A recursive subtype constraint of the form $\tau.u \subtype \tau.v$
\item A subtype constraint of the form $\tau.u \subtype \overline{\kappa}$ or
$\overline{\kappa} \subtype \tau.u$ where $\overline{\kappa}$ is a
type constant.
\end{itemize}
We will call a constraint set $\mathcal{C}'$ a {\em simplification} of $\mathcal{C}$ if 
$\mathcal{C}' \proves c$ for every interesting constraint $c$, such that $\mathcal{C} \proves c$. Since both $\mathcal{C}$ and $\mathcal{C}'$ entail the same set of constraints on $\tau$,
it is valid to replace $\mathcal{C}$ with $\mathcal{C}'$ in any valid type scheme for $\tau$.
\label{simplification-def}
\end{definition}

Simplification heuristics for set-constraint systems were studied by
\citet{fahndrich1996making}; our simplification algorithm encompasses
all of these heuristics.

\subsection{Unconstrained Pushdown Systems}
\label{pds-section}

The constraint-simplification algorithm works on a constraint set $\mathcal{C}$ by building a
pushdown system $\mathcal{P}_\mathcal{C}$ whose transition sequences represent valid derivations of subtyping
judgements. We briefly review pushdown systems and some necessary generalizations here.

\begin{definition}
An {\em unconstrained pushdown system} is a triple $\mathcal{P} = (\mathcal{V}, \Sigma, \Delta)$ where
$\mathcal{V}$ is the set of {\em control locations}, $\Sigma$ is the set of {\em stack symbols},
and $\Delta \subseteq (\mathcal{V} \times \Sigma^*)^2$ is a (possibly infinite) set of {\em transition rules}.
We will denote a transition rule by $\pdsrule{X}{u}{Y}{v}$ where $X,Y \in \mathcal{V}$ and
$u,v \in \Sigma^*$. We define the set of {\em configurations} to be $\mathcal{V} \times \Sigma^*$.
In a configuration $(p,w)$, $p$ is called the {\em control state} and $w$ the {\em stack state}.
\end{definition}

Note that we require neither the set of stack symbols, nor the set of transition rules,
to be finite.  This freedom is required to model the derivation $\textsc{S-Pointer}$
of \autoref{deduction-rules}, which
corresponds to an infinite set of transition rules.

\begin{definition}
An unconstrained pushdown system $\mathcal{P}$ determines a {\em transition relation} $\rewrite$ on
the set of configurations:
$(X,w) \rewrite (Y,w')$ if there is a suffix $s$ and a rule $\pdsrule{X}{u}{Y}{v}$,
such that $w = us$ and $w' = vs$. The transitive closure of $\rewrite$ is denoted $\rewrites$.
\end{definition}

With this definition, we can state the primary theorem behind our simplification algorithm.
\begin{theorem}
\label{mod-sat}
Let $\mathcal{C}$ be a constraint set and $\mathcal{V}$ a set of base type variables. Define a
subset $S_\mathcal{C}$ of $(\mathcal{V} \cup \Sigma)^* \times (\mathcal{V} \cup \Sigma)^*$
by $(Xu, Yv) \in S_\mathcal{C}$ if and only if $\mathcal{C} \proves X.u \subtype Y.v$.
Then $S_\mathcal{C}$ is a regular set, and an automaton $Q$ to recognize $S_\mathcal{C}$ can
be constructed in $O(|\mathcal{C}|^3)$ time.
\end{theorem}
\begin{proof}The basic idea is to treat each $X.u \subtype Y.v \in \mathcal{C}$
as a transition rule $\pdsrule{X}{u}{Y}{v}$ in the pushdown system $\mathcal{P}$.  In addition,
we add control states $\stStart, \stEnd$ with transitions $\pdsrule{\stStart}{X}{X}{\varepsilon}$
and $\pdsrule{X}{\varepsilon}{\stEnd}{X}$ for each $X \in \mathcal{V}$.
For the moment,
assume that (1) all labels are covariant, and (2) the rule $\textsc{S-Pointer}$ is ignored.
By construction, $(\stStart, Xu) \rewrites (\stEnd, Yv)$ in $\mathcal{P}$ if and only if $\mathcal{C} \proves X.u \subtype Y.v$.  A theorem of B\"uchi \cite{richard1964regular} ensures that  for any two control states $A$ and $B$ in a standard (not unconstrained) pushdown system, the set of all pairs $(u,v)$
with $(A, u) \rewrites (B, v)$ is a regular language; \citet{caucal1992regular} gives a saturation
algorithm that constructs an automaton to recognize this language.

In the full proof, we add two novelties: first, we support contravariant stack symbols by
encoding variance data into the control states and transition rules. The second novelty
involves the rule $\textsc{S-Pointer}$; this rule is problematic since the natural
encoding would result in infinitely many transition rules. We extend Caucal's construction to lazily instantiate all necessary applications of
$\textsc{S-Pointer}$ during saturation.  For details, see \aautoref{construction}{Appendix D of \citep{arxiv}}.
\end{proof}

Since $\mathcal{C}$ will usually entail an infinite number of constraints,
this theorem is particularly useful: it tells us that the full set of
constraints entailed by $\mathcal{C}$ has a finite encoding by an automaton $Q$.
Further manipulations on the constraint closure, such as efficient minimization, can be carried out on $Q$. By restricting the transitions to and from $\stStart$ and $\stEnd$, the same algorithm
is used to eliminate type variables, producing the desired constraint simplifications.

\subsection{Overall Complexity of Inference}
The saturation algorithm used to perform constraint-set simplification and type-scheme construction
is, in the worst case, cubic in the number of subtype constraints to simplify.  Since some well-known pointer analysis methods also have cubic complexity (such as \citet{andersen1994program}), it is reasonable to wonder if Retypd's ``points-to free'' analysis really offers a benefit over a type-inference system built on top of points-to analysis data.

To understand where Retypd's efficiencies are found, first consider the $n$ in $O(n^3)$.  Retypd's core saturation algorithm is cubic in the number of subtype constraints; due to the simplicity of machine-code instructions, there is roughly one subtype constraint generated per instruction.  Furthermore, Retypd applies constraint simplification on each procedure in isolation to eliminate the procedure-local type variables, resulting in constraint sets that only relate procedure formal-ins, formal-outs, globals, and type constants. In practice, these simplified constraint sets are small.

Since each procedure's constraint set is simplified independently, the $n^3$ factor is controlled by the largest procedure size, not the overall size of the binary.  By contrast, source-code points-to analysis such as Andersen's are generally cubic in the overall number of pointer variables, with exponential duplication of variables depending on the call-string depth used for context sensitivity.  The situation is even more difficult for machine-code points-to analyses such as VSA, since there is no syntactic difference between a scalar and a pointer in machine code. In effect, every program variable must be treated as a potential pointer.

On our benchmark suite of real-world programs, we found that execution time for Retypd scales slightly below $O(N^{1.1})$, where $N$ is the number of program instructions.  The following back-of-the-envelope calculation can heuristically explain much of the disparity between the $O(N^3)$ theoretical complexity and the $O(N^{1.1})$ measured complexity.  On our benchmark suite, the maximum procedure size $n$ grew roughly like $n \approx N^{2/5}$. We could then expect that a per-procedure analysis would perform worst when the program is partitioned into $N^{3/5}$ procedures of size $N^{2/5}$.
On such a program, a per-procedure $O(n^k)$ analysis may be expected to behave more like an $O(N^{3/5} \cdot (N^{2/5})^k) = O(N^{({3 + 2k})/{5}})$ analysis overall.  In particular, a per-procedure cubic analysis like Retypd could be expected to scale like a global $O(N^{1.8})$ analysis.  The remaining differences in observed versus theoretical execution time can be explained by the facts that real-world constraint graphs do not tend to exercise the simplification algorithm's worst-case behavior, and that the distribution of procedure sizes is heavily weighted towards small procedures.

\section {Evaluation}
\label{sec-evaluation}

\begin{figure}
  \hrule
\begin{lstlisting}[frame=none]
INT_PTR Proto_EnumAccounts(WPARAM wParam,
                           LPARAM lParam)
{
  *( int* )wParam = accounts.getCount();
  *( PROTOACCOUNT*** )lParam =
      accounts.getArray();
  return 0;
}
\end{lstlisting}
\caption{Ground-truth types declared in the original source do not necessarily reflect program semantics. Example from \texttt{miranda32}.}
\label{miranda}
\end{figure}

\pgfkeys{/pgf/number format/set thousands separator={}}

\newcommand{\xcsvkey}{Avg ptr accuracy}
\newcommand{\xdescr}{Accuracy}
\newcommand{\ycsvkey}{Avg conservativeness}
\newcommand{\ydescr}{Conservativeness}
\newcommand{\zcsvkey}{Avg interval size}
\newcommand{\zdescr}{Mean interval size}

\newcommand{\csvfile}[1]{data#1.csv}

\pgfdeclareplotmark{fivestar}
{  \pgfpathmoveto{\pgfqpointpolar{90}{\pgfplotmarksize}}
  \pgfpathlineto{\pgfqpointpolar{126}{0.4\pgfplotmarksize}}
  \pgfpathlineto{\pgfqpointpolar{162}{\pgfplotmarksize}}
  \pgfpathlineto{\pgfqpointpolar{198}{0.4\pgfplotmarksize}}
  \pgfpathlineto{\pgfqpointpolar{234}{\pgfplotmarksize}}
  \pgfpathlineto{\pgfqpointpolar{270}{0.4\pgfplotmarksize}}
  \pgfpathlineto{\pgfqpointpolar{306}{\pgfplotmarksize}}
  \pgfpathlineto{\pgfqpointpolar{342}{0.4\pgfplotmarksize}}
  \pgfpathlineto{\pgfqpointpolar{378}{\pgfplotmarksize}}
  \pgfpathlineto{\pgfqpointpolar{414}{0.4\pgfplotmarksize}}
  \pgfpathclose
  \pgfusepathqfillstroke
}

\subsection{Implementation}
Retypd is implemented as a module within CodeSurfer for Binaries.
By leveraging the multi-platform disassembly capabilities of CodeSurfer, it
can operate on x86, x86-64, and ARM code. We performed the 
evaluation using minimal analysis settings, disabling value-set analysis (VSA) but
computing affine relations between the stack and frame pointers.
Enabling additional CodeSurfer phases such as VSA can greatly improve the 
reconstructed IR, at the expense of increased analysis time.

Existing type-inference algorithms such as TIE \citep{tie} and SecondWrite \citep{2ndwrite}
require some modified form of VSA to resolve points-to data.  Our approach shows that
high-quality types can be recovered in the absence of points-to information,
allowing type inference to proceed even when computing points-to data is too unreliable
or expensive.

\subsection{Evaluation Setup}

\begin{figure}
\centering
\begin{tabular}{@{}llr@{}}
\toprule
Benchmark &  Description & Instructions \\
\midrule
\multicolumn{2}{l}{\em CodeSurfer benchmarks} & \\[0.75ex]
{libidn} & Domain name translator & 7K \\
{Tutorial00} & Direct3D tutorial & 9K\\
{zlib} & Compression library & 14K\\
{ogg} & Multimedia library & 20K\\
{distributor} & UltraVNC repeater & 22K \\
{libbz2} & BZIP library, as a DLL & 37K \\
{glut} & The {\tt glut32.dll} library & 40K \\
{pngtest} & A test of libpng & 42K \\
{freeglut} & The {\tt freeglut.dll} library & 77K \\
{miranda} & IRC client & 100K \\
{XMail} & Email server & 137K \\
{yasm} & Modular assembler & 190K \\
{python21} & Python 2.1 & 202K \\
{quake3} & Quake 3 & 281K \\
{TinyCad} & Computed-aided design & 544K \\
{Shareaza} & Peer-to-peer file sharing & 842K \\
\midrule
\multicolumn{2}{l}{\em SPEC2006 benchmarks} & \\[0.75ex]
{470.lbm} & Lattice Boltzmann Method & 3K \\
{429.mcf} & Vehicle scheduling & 3K\\
{462.libquantum} & Quantum computation & 11K \\
{401.bzip2} & Compression & 13K \\
{458.sjeng} & Chess AI & 25K \\
{433.milc} & Quantum field theory & 28K\\
{482.sphinx3} & Speech recognition & 43K \\
{456.hmmer} & Protein sequence analysis & 71K \\
{464.h264ref} & Video compression & 113K \\
{445.gobmk} & GNU Go AI & 203K \\
{400.perlbench} & Perl core & 261K \\
{403.gcc} & C/C++/Fortran compiler & 751K \\
\end{tabular}
\caption{Benchmarks used for evaluation. All binaries were compiled from source using optimized release configurations. The SPEC2006 benchmarks were chosen to match the benchmarks used to evaluate SecondWrite \citep{2ndwrite}.}
\label{test-table}
\end{figure}

Our benchmark suite consists of 160 32-bit x86 binaries for both Linux and Windows, compiled with a
variety of  {\verb|gcc|} and Microsoft Visual C/C++ versions.  The benchmark suite includes a mix of
executables, static libraries, and DLLs.  The suite includes the same \verb|coreutils| and SPEC2006
benchmarks used to evaluate REWARDS, TIE, and SecondWrite \citep{rewards,tie,2ndwrite}; additional benchmarks came from a standard suite of real-world programs used for precision and performance testing of CodeSurfer for Binaries.  All binaries were built with optimizations enabled and debug information disabled.

Ground truth is provided by separate copies of the binaries that have been built with the same settings, but with debug information included (DWARF on Linux, PDB on Windows). We used IdaPro \citep{idapro} to read the debug information, which allowed us to use the same scripts for collecting ground-truth types from both DWARF and PDB data.

All benchmarks were evaluated on a 2.6 GHz Intel Xeon E5-2670 CPU, running on a single logical core.
RAM utilization by CodeSurfer and Retypd combined was capped at 10GB.

Our benchmark suite includes the individual binaries in \autoref{test-table} as well
as the collections of related binaries shown in \autoref{cluster-table}.  We found that 
programs from a single collection tended to share a large amount of common code, leading
to highly correlated benchmark results. For example, even though
the {\verb|coreutils|} benchmarks include many tools with very disparate purposes, all of the
tools make use of a large, common set of statically linked utility routines. Over 80\% of the
{\verb|.text|} section in {\verb|tail|} consists of such routines; for {\verb|yes|}, the number is over 99\%. The common code and specific idioms appearing in {\verb|coreutils|} make it a particularly low-variance benchmark suite.

In order to avoid over-representing these program collections in our results, we
treated these collections as clusters in the data set.  For each cluster, we computed the
average of each metric over the cluster, then inserted the average as a
single data point to the final data set.  Because Retypd performs well on many clusters,
this averaging procedure tends to reduce our overall precision and conservativeness measurements.
Still, we believe that it gives a less biased depiction of the algorithm's expected real-world behavior than does an average over all benchmarks.

\renewcommand{\xcsvkey}{Name}
\renewcommand{\xdescr}{Name}
\begin{figure*}
\begin{minipage}[t]{0.45\linewidth}
\[
\begin{tikzpicture}[scale=0.9]

\begin{groupplot}[
        group style={
              group name=eval plots,
              group size=3 by 1,
              horizontal sep=0pt,
        },
        /pgf/bar width=0.7em,
        footnotesize,
        tickpos=left,
	visualization depends on={value \thisrow{Name} \as \lab},
	every node near coord/.append style={font={\small}},
	grid=major,
        major grid style={gray},
        xmajorgrids=false,
        xminorgrids=false,
        height=0.8\textwidth,
        ytick align=outside,
        xtick align=outside,
        x=0.1\textwidth,
        x tick label style={rotate=45,anchor=east},
        xtick=data,
        ytick={0,0.5,1,1.5,2},
        enlarge x limits={abs={1.5em}},
	ymin=0,
	ymax=2,
        ybar=0pt,
    ]

\nextgroupplot[
        title=coreutils,
        symbolic x coords={REWARDS-c*,TIE*,TIE,Retypd},
	]

\renewcommand{\ycsvkey}{Avg dist}
\renewcommand{\ydescr}{Distance}

	\addplot[fill=blue!75!green] table [x={\xcsvkey},y={\ycsvkey},col sep=comma]{compare-c.csv};

\renewcommand{\ycsvkey}{Avg interval size}
\renewcommand{\ydescr}{Interval size}

        \addplot[fill=orange!75!white]		 table [x={\xcsvkey},y={\ycsvkey},z={\zcsvkey},col sep=comma]{compare-c.csv};

\nextgroupplot[
        title=SPEC2006,
        symbolic x coords={SecondWrite,Retypd},
        yticklabels={,,},
        ymajorticks=false,
        enlarge x limits={abs={1.5em}},
	]

\renewcommand{\ycsvkey}{Avg dist}
\renewcommand{\ydescr}{Distance}

	\addplot[fill=blue!75!green] table [x={\xcsvkey},y={\ycsvkey},col sep=comma]{compare-s.csv};

\renewcommand{\ycsvkey}{Avg interval size}
\renewcommand{\ydescr}{Interval size}

        \addplot[fill=orange!75!white]		 table [x={\xcsvkey},y={\ycsvkey},z={\zcsvkey},col sep=comma]{compare-s.csv};

\nextgroupplot[
        title=All,
        symbolic x coords={Retypd},
        yticklabel pos=right,
        ytick pos=right,
        x=0.02\textwidth,
	]

\renewcommand{\ycsvkey}{Avg dist}
\renewcommand{\ydescr}{Distance}

	\addplot[fill=blue!75!green] table [x={\xcsvkey},y={\ycsvkey},col sep=comma]{compare-o.csv};

\renewcommand{\ycsvkey}{Avg interval size}
\renewcommand{\ydescr}{Interval size}

        \addplot[fill=orange!75!white]		 table [x={\xcsvkey},y={\ycsvkey},z={\zcsvkey},col sep=comma]{compare-o.csv};

  \end{groupplot}

\node at (eval plots c2r1.north west) [inner sep=0pt,anchor=south, yshift=4ex] {
   \begin{tikzpicture}
        \draw[fill=blue!75!green] (0cm,-0.1cm) rectangle (0.5cm,0.1cm);
   \end{tikzpicture} {\scriptsize Distance to source type}
   \quad
   \begin{tikzpicture}
        \draw[fill=orange!75!white] (0cm,-0.1cm) rectangle (0.5cm,0.1cm);
   \end{tikzpicture} {\scriptsize Interval size}
};

\node[font=\scriptsize] at (6,-1.3) {$*$ Dynamic};

\end{tikzpicture}
\]
\caption{Distance to ground-truth types and size of the interval between inferred upper and lower bounds.
Smaller distances represent more accurate types; smaller interval sizes represent increased
confidence.}
\label{dist-interval}
\end{minipage}
\hspace{3em}
\begin{minipage}[t]{0.45\linewidth}
\renewcommand{\xcsvkey}{Name}
\renewcommand{\xdescr}{Name}
\[
\begin{tikzpicture}[scale=0.9]
\begin{groupplot}[
        group style={
              group name=eval plots,
              group size=3 by 1,
              horizontal sep=0pt,
        },
        /pgf/bar width=0.7em,
        footnotesize,
        tickpos=left,
	visualization depends on={value \thisrow{Name} \as \lab},
	every node near coord/.append style={font={\small}},
	grid=major,
        major grid style={gray},
        xmajorgrids=false,
        xminorgrids=false,
        height=0.8\textwidth,
        ytick align=outside,
        xtick align=outside,
        x=0.1\textwidth,
        x tick label style={rotate=45,anchor=east},
        xtick=data,
        ytick={70,80,90,100},
        yticklabel={\pgfmathparse{\tick}\pgfmathprintnumber{\pgfmathresult}\%},
        axis y discontinuity=crunch,
        enlarge x limits={abs={1.5em}},
	ymin=60,
	ymax=100,
        ybar=0pt,
        legend columns=-1,
        legend style={draw=none, /tikz/every even column/.append style={column sep=5pt}},
        legend image code/.code={             \draw[#1] (0cm,-0.1cm) rectangle (0.5cm,0.1cm);
             },
    ]

\nextgroupplot[
        title=coreutils,
        symbolic x coords={REWARDS-c*,TIE*,TIE,Retypd},
	]

\renewcommand{\ycsvkey}{Avg conservativeness}
\renewcommand{\ydescr}{Conservativeness}

	\addplot[fill=yellow] table [x={\xcsvkey},y={\ycsvkey},col sep=comma]{compare-c.csv};

\renewcommand{\ycsvkey}{Avg ptr accuracy}
\renewcommand{\ydescr}{Pointer accuracy}

        \addplot[fill=red!75!blue]		 table [x={\xcsvkey},y={\ycsvkey},z={\zcsvkey},col sep=comma]{compare-c.csv};

\nextgroupplot[
        title=SPEC2006,
        symbolic x coords={SecondWrite,Retypd},
        yticklabels={,,},
        ymajorticks=false,
        enlarge x limits={abs={1.5em}},
	]

\renewcommand{\ycsvkey}{Avg conservativeness}
\renewcommand{\ydescr}{Conservativeness}

	\addplot[fill=yellow] table [x={\xcsvkey},y={\ycsvkey},col sep=comma]{compare-s.csv};

\renewcommand{\ycsvkey}{Avg ptr accuracy}
\renewcommand{\ydescr}{Pointer accuracy}

        \addplot[fill=red!75!blue]		 table [x={\xcsvkey},y={\ycsvkey},z={\zcsvkey},col sep=comma]{compare-s.csv};

\nextgroupplot[
        title=All,
        symbolic x coords={Retypd},
        yticklabel={\pgfmathparse{\tick}\pgfmathprintnumber{\pgfmathresult}\%},
        yticklabel pos=right,
        ytick pos=right,
        x=0.02\textwidth,
        legend to name=grpleg,             
	]

\renewcommand{\ycsvkey}{Avg conservativeness}
\renewcommand{\ydescr}{Conservativeness}

	\addplot[fill=yellow] table [x={\xcsvkey},y={\ycsvkey},col sep=comma]{compare-o.csv};
        \addlegendentry{Conservativeness}
        
\renewcommand{\ycsvkey}{Avg ptr accuracy}
\renewcommand{\ydescr}{Pointer accuracy}

        \addplot[fill=red!75!blue]		 table [x={\xcsvkey},y={\ycsvkey},z={\zcsvkey},col sep=comma]{compare-o.csv};
        \addlegendentry{Pointer accuracy}
        
  \end{groupplot}

\node at (eval plots c2r1.north west) [inner sep=0pt,anchor=south, yshift=4ex] {
   \begin{tikzpicture}
        \draw[fill=yellow] (0cm,-0.1cm) rectangle (0.5cm,0.1cm);
   \end{tikzpicture} {\scriptsize Conservativeness}
   \quad
   \begin{tikzpicture}
        \draw[fill=red!75!blue] (0cm,-0.1cm) rectangle (0.5cm,0.1cm);
   \end{tikzpicture} {\scriptsize Pointer accuracy}
};

\draw [ultra thick, white, decoration={snake, amplitude=1pt, segment length=7pt}, decorate] (0.2,0.4) -- (2.7,0.4);
\draw [ultra thick, white, decoration={snake, amplitude=1pt, segment length=7pt}, decorate] (3.1,0.4) -- (4.4,0.4);
\draw [ultra thick, white, decoration={snake, amplitude=1pt, segment length=7pt}, decorate] (4.8,0.4) -- (5.8,0.4);

\node[font=\scriptsize] at (6,-1.3) {$*$ Dynamic};

\end{tikzpicture}
\]
\caption{Conservativeness and pointer accuracy metric. Perfect type reconstruction would be 100\%
conservative and match on 100\% of pointer levels. Note that the $y$ axis begins at 70\%.}
\label{conserv-ptr-acc}

\end{minipage}
\end{figure*}

\subsection{Sources of Imprecision}

Although Retypd is built around a sound core of constraint simplification and solving,
there are several ways that imprecision can occur. As described in \autoref{dis-fail},
disassembly failures can lead to unsound constraint generation. Second, the heuristics
for converting from sketches to C types are lossy by necessity. Finally, we treat the source types as ground truth, leading to ``failures'' whenever Retypd recovers an accurate type that
does not match the original program---a common situation with type-unsafe source code.

A representative example of this last source of imprecision appears in \autoref{miranda}. This source code belongs to the {\verb|miranda32|} IRC client, which uses a plugin-based architecture; most of {\verb|miranda32|}'s functionality is implemented by ``service functions'' with the fixed signature {\verb|int ServiceProc(WPARAM,LPARAM)|}. The types {\verb|WPARAM|} and {\verb|LPARAM|} are used in certain Windows APIs for generic 16- and 32-bit values.  The two parameters are immediately cast to other types in the body of the service functions, as in \autoref{miranda}.

\subsection{{\tt const} Correctness}
\label{const-recovery-stats}

As a side-effect of separately modeling $\mathsf{.load}$ and $\mathsf{.store}$ capabilities,
Retypd is easily able to recover information about how pointer parameters are used for
input and/or output.  We take this into account when converting
sketches to C types; if a function's sketch includes $.\mathsf{in}_L .\mathsf{load}$ but not
$.\mathsf{in}_L . \mathsf{store}$ then we annotate the parameter at $L$ with {\verb|const|},
as in \autoref{sketch-sketch} and \autoref{type-recovery-example}.
Retypd appears to be the first machine-code type-inference system to infer {\verb|const|}
annotations directly.

On our benchmark suite, we found that 98\% of parameter {\verb|const|} annotations in the original source
code were recovered by Retypd.  Furthermore, Retypd inferred {\verb|const|} annotations on many
other parameters; unfortunately, since most C and C++ code does not use {\verb|const|} in every
possible situation, we do not have a straightforward way to detect how many of Retypd's additional
{\verb|const|} annotations are correct.

Manual inspection of the missed {\verb|const|} annotations shows that most instances are due to
imprecision when analyzing one or two common statically linked library functions.  This imprecision then propagates
outward to callers, leading to decreased {\verb|const|} correctness overall.  Still, we believe the 98\% recovery rate shows that Retypd offers a useful approach to {\verb|const|} inference.

\subsection{Comparisons to Other Tools}
We gathered results over several metrics that have been used to evaluate SecondWrite,
TIE, and REWARDS.  These metrics were defined by \citet{tie} and are briefly reviewed here.

TIE infers upper and lower bounds on each type variable, with the bounds belonging to a
lattice of C-like types.  The lattice is naturally stratified into levels, with the distance between
two comparable types roughly being the difference between their levels in the lattice, with a
maximum distance of 4.  A recursive
formula for computing distances between pointer and structural types is also used.  TIE also determines a policy that selects between the upper and lower bounds on a type variable for the final displayed type.

TIE considers three metrics based on this lattice: the conservativeness rate, the interval size, and
the distance.
A type interval is {\em conservative} if the interval bounds overapproximate the declared type of a variable.
The {\em interval size} is the lattice distance from the upper to the lower bound on a type variable.
The {\em distance} measures the lattice distance from the final displayed type to the ground-truth type.
REWARDS and SecondWrite both use unification-based algorithms, and have been evaluated using the same
TIE metrics. The evaluation of REWARDS using TIE metrics appears in \citet{tie}.

\noindent
{\bf Distance and interval size:}
Retypd shows substantial improvements over other approaches in the distance and interval-size metrics,
indicating that it generates more accurate types with less uncertainty.  The mean distance to
the ground-truth type was 0.54 for Retypd, compared to 1.15 for dynamic TIE, 1.53 for REWARDS, 1.58 for static TIE,
and 1.70 for SecondWrite.  The mean
interval size shrunk to 1.2 with Retypd, compared to 1.7 for SecondWrite and 2.0 for TIE.

\noindent
{\bf Multi-level pointer accuracy:}
\citet{2ndwrite} also introduced a multi-level pointer-accuracy rate that attempts to
quantify how many ``levels'' of pointers were correctly inferred.
\cbstart
On SecondWrite's benchmark
suite, Retypd attained a mean multi-level pointer accuracy of 91\%, compared with SecondWrite's reported 73\%.
Across all benchmarks, Retypd averages 88\% pointer accuracy.
\cbend

\noindent
{\bf Conservativeness:}
The best type system would have a high conservativeness rate (few unsound decisions) coupled with
a low interval size (tightly specified results) and low distance (inferred types are close to ground-truth types).
In each of these metrics, Retypd performs about as well or better than existing approaches.  Retypd's mean conservativeness
rate is 95\%, compared to 94\% for TIE.  But note that TIE was evaluated only on
{\verb|coreutils|};
on that cluster, Retypd's conservativeness was 98\%.  
SecondWrite's overall conservativeness is 96\%, measured on a subset of the SPEC2006 benchmarks;
Retypd attained a slightly lower 94\% on this subset.

It is interesting to note that Retypd's conservativeness rate on {\verb|coreutils|} is comparable
to that of REWARDS, even though REWARDS' use of dynamic execution traces suggests it would be more
conservative than a static analysis by virtue of only generating feasible type constraints.

  \begin{figure*}
\centering
\begin{tabular}{@{}lllrccccc@{}}
\toprule
Cluster & Count & Description & Instructions & Distance & Interval & Conserv. & Ptr. Acc. & Const \\
\midrule
freeglut-demos & 3 & freeglut samples & 2K & \fbox{0.66} & \fbox{1.49} & 97\% & \fbox{83\%} & 100\% \\
coreutils & 107 & GNU coreutils 8.23 & 10K & 0.51 & 1.19 & 98\% & \fbox{82\%} & \fbox{96\%} \\
vpx-d & 8 & VP$x$ decoders & 36K & \fbox{0.63} & \fbox{1.68} & 98\% & 92\% & 100\% \\
vpx-e & 6 & VP$x$ encoders & 78K & \fbox{0.63} & \fbox{1.53} & 96\% & 90\% & 100\% \\
sphinx2 & 4 & Speech recognition & 83K & 0.42 & 1.09 & \fbox{94\%} & 91\% & 99\% \\
putty & 4 & SSH utilities & 97K & 0.51 & 1.05 & \fbox{94\%} & \fbox{86\%} & 99\% \\
\midrule
\multicolumn{3}{r}{Retypd, as reported} & & 0.54 & 1.20 & 95\% & 88\% & 98\% \\
\multicolumn{3}{r}{Retypd, without clustering} & & 0.53 & 1.22 & 97\% & 84\% & 97\% \\
\end{tabular}
\caption{Clusters in the benchmark suite.  For each metric, the average over the cluster is given. If a cluster average is worse than Retypd's overall average for a certain metric, a box is drawn around the entry.}
\label{cluster-table}
  \end{figure*}

\subsection{Performance}
\begin{figure*}
\renewcommand{\xcsvkey}{size}
\renewcommand{\xdescr}{Program size (number of CFG nodes)}
\renewcommand{\ycsvkey}{time}
\renewcommand{\ydescr}{Type-inference time (seconds)}
\begin{minipage}[b]{0.45\linewidth}\[
\begin{tikzpicture}[scale=0.9]

  \begin{axis}[
        xmin=1000,
        xmax=1000000,
        ymin=0,
        ymax=3000,
        xmode=log,
        ymode=log,
        scaled ticks=true,
        extra y ticks={3000},
        xticklabels={,,},
        extra x ticks={1000,10000,100000,1000000},
        extra x tick labels={1K, 10K, 100K, 1M},
        log ticks with fixed point,
        major grid style={dotted},
	xlabel={\xdescr},
	ylabel={\ydescr},
	]
	
        \draw[color=red,very thick] (1000,1.427) -- (1000000,2812);
        \addplot[scatter,only marks,nodes near coords={},fill={blue!25!white},mark options={draw=blue,thick,mark size=3pt}] table [x={\xcsvkey},y={\ycsvkey},col sep=comma]{\csvfile{-timing}};

        \end{axis}
\end{tikzpicture}\]
\caption{Type-inference time on benchmarks. The line indicates the best-fit exponential
$t = 0.000725 \cdot N^{1.098}$, demonstrating slightly superlinear real-world scaling
behavior. The coefficient of determination is $R^2 = 0.977$.}
\label{performance-fig}
\end{minipage}
\renewcommand{\xcsvkey}{Size}
\renewcommand{\xdescr}{Program size (number of CFG nodes)}
\renewcommand{\ycsvkey}{Memory}
\renewcommand{\ydescr}{Type-inference memory usage}
\hspace{2em}
\begin{minipage}[b]{0.45\linewidth}
\[
\begin{tikzpicture}[scale=0.9]

  \begin{axis}[
        xmin=1000,
        xmax=1000000,
        ymin=0,
        ymax=4800,
        xmode=log,
        ymode=log,
        scaled ticks=true,
        xticklabels={,,},
        yticklabels={,,},
        extra x ticks={1000,10000,100000,1000000},
        extra x tick labels={1K, 10K, 100K, 1M},
        extra y ticks={50,100,500,1000,2000,4000},
        extra y tick labels={50MB, 100MB, 500MB, 1GB, 2GB, 4GB},
        log ticks with fixed point,
        major grid style={dotted},
	xlabel={\xdescr},
	ylabel={\ydescr},
	]
	
        \draw[color=red,very thick] (1000,12.8) -- (1000000,4437);

        \addplot[scatter,only marks,nodes near coords={},fill={blue!25!white},mark options={draw=blue,thick,mark size=3pt}] table [x={\xcsvkey},y={\ycsvkey},col sep=comma]{\csvfile{-mem}};

        \end{axis}
\end{tikzpicture}\]
\caption{Type-inference memory usage on benchmarks. The line indicates the best-fit exponential $m = 0.037 \cdot N^{0.846}$, with coefficient of determination $R^2 = 0.959$.}% \vspace{2.4ex} }
\label{performance-mem-fig}
\end{minipage}
\end{figure*}

Although the core simplification algorithm of Retypd has cubic worst-case complexity,
it only needs to be applied on a per-procedure basis. This suggests that the real-world scaling behavior will depend on the distribution of procedure sizes, not on the whole-program size.

In practice, \autoref{performance-fig} suggests that Retypd gives nearly linear performance over the benchmark suite, which ranges
in size from 2K to 840K instructions.  To measure Retypd's performance, we used numerical
regression to find the best-fit model $T = \alpha N^\beta$ relating execution time $T$ to
program size $N$.  This results in the relation $T = 0.000725 \cdot N^{1.098}$ with coefficient
of determination $R^2 = 0.977$, suggesting that nearly 98\% of the variation in performance data can
be explained by this model.  In other words, on real-world programs Retypd demonstrates nearly linear scaling of execution time. The cubic worst-case per-procedure behavior of the constraint solver does not translate to cubic behavior overall.
Similarly, we found that the sub-linear model $m = 0.037 \cdot N^{0.846}$ explains 96\% of the memory usage in Retypd.

\noindent{\bf Note.}
The regressions above were performed by numerically fitting exponential models in $(N,T)$ and $(N,m)$ space, rather than analytically fitting linear models in $\log$-$\log$ space.  Our models then minimize the error in the predicted values of $T$ and $m$, rather than minimizing errors in $\log T$ or $\log m$. Linear regression in $\log$-$\log$ space results in the less-predicative models $T = 0.0003 \cdot N^{1.14}$ ($R^2 = 0.92$) and $m = 0.08 \cdot N^{0.77}$ ($R^2 = 0.88$).

The constraint-simplification workload of Retypd would be straightforward to parallelize
over the directed acyclic graph of strongly-connected components in the callgraph, further
reducing the scaling constant. If a reduction in memory usage is required, Retypd could swap
constraint sets to disk; the current implementation keeps all working sets in RAM.

\section {Related Work}

{\bf Machine-code type recovery:}
Hex-Ray's reverse engineering tool IdaPro \citep{idapro} is an early example of type
reconstruction via static analysis.  The exact algorithm is proprietary, but it appears
that IdaPro propagates types through unification from library functions of known signature,
halting the propagation when a type conflict appears.  IdaPro's reconstructed IR is 
relatively sparse, so the type propagation fails to produce useful information in many
common cases, falling back to the default {\verb|int|} type.  However, the analysis is very fast.

SecondWrite \citep{2ndwrite} is an interesting approach to static IR reconstruction
with a particular emphasis on scalability.  The authors combine a best-effort VSA variant for
points-to analysis with a unification-based type-inference engine.  Accurate types
in SecondWrite depend on high-quality points-to data; the authors note that this can
cause type accuracy to suffer on larger programs.
In contrast, Retypd is not dependent on points-to data for type recovery and makes
use of subtyping rather than unification for increased precision.

TIE \citep{tie} is a static type-reconstruction tool used as part of Carnegie Mellon
University's binary-analysis platform (BAP).  TIE was the first machine-code type-inference system to track subtype constraints and explicitly maintain upper
and lower bounds on each type variable.  As an abstraction of the
C type system, TIE's type lattice is relatively simple; missing features, such as
recursive types, were later identified by the authors as an important target for
future research \citep{schwartz}.

HOWARD \citep{howard} and REWARDS \citep{rewards} both take a dynamic approach, generating type
constraints from execution traces.  Through a comparison with HOWARD, the creators of TIE showed
that static type analysis can produce higher-precision types than dynamic type analysis, 
though a small penalty must be paid in conservativeness of constraint-set generation.
TIE also showed that type systems designed for static analysis can be easily modified to work on
dynamic traces; we expect the same is true for Retypd, though we have not yet
performed these experiments.

Most previous work on machine-code type recovery, including TIE and SecondWrite, either
disallows recursive types or only supports recursive types
by combining type-inference results with a points-to oracle.
For example, to infer that $x$ has a
the type {\verb|struct S { struct S *, ...}*|} in a unification-based approach
like SecondWrite,
first we must have resolved that $x$ points to some memory region $M$, that $M$ admits a 4-byte abstract
location $\alpha$ at offset 0, and that the type of $\alpha$ should be unified with the type of $x$.
If pointer analysis
has failed to compute an explicit memory region pointed to by $x$, it will not be
possible to determine the type of $x$ correctly.
The complex interplay between type inference, points-to analysis, and abstract-location
delineation leads to a relatively fragile method for inferring recursive types.
In contrast, our type system can infer recursive types
even when points-to facts are completely absent.

\citet{robbins2013theory} developed
an SMT solver equipped with a theory of rational trees and applied it to
type reconstruction. Although this allows for recursive types, the lack of
subtyping and the performance of the SMT solver make it difficult to
scale this approach to real-world binaries. Except for test cases on the
order of 500 instructions, precision of the recovered types was not assessed.

\smallskip
\noindent
{\bf Related type systems:}
The type system used by Retypd is related to the recursively constrained types 
(rc types) of \citet*{eifrig}.
Retypd generalizes the rc type system by building up all types using flexible
records;
even the function-type constructor $\to$, taken as fundamental in the rc type system, is decomposed into a record with $\mathsf{in}$ and $\mathsf{out}$ fields.
This allows Retypd to operate without the knowledge
of a fixed signature from which type constructors are drawn, which is essential for analysis of
stripped machine code. 

The use of CFL reachability to perform polymorphic subtyping first appeared in \citet{rehof2001type},
extending previous work relating simpler type systems
to graph reachability \citep{agesen94, palsberg1995type}.  Retypd continues by adding
type-safe handling of pointers and a simplification algorithm that allows us to compactly
represent the type scheme for each function.

CFL reachability has also been used to
extend the type system of Java \citep{greenfieldboyce2007type} and C++ \citep{foster2006flow} with support
for additional type qualifiers.  Our reconstructed {\verb|const|} annotations can be seen
 as an instance of this idea, although our qualifier inference is not separated from
type inference.

To the best of our knowledge, no prior work has applied 
polymorphic type systems with subtyping to machine code.
 
\section {Future Work}

One interesting avenue for future research could come from the application of dependent
and higher-rank type systems to machine-code type inference, although we
rapidly approach the frontier where type inference is undecidable.  A natural example of
dependent types appearing in machine code is {{\verb|malloc|}}, which could be typed as 
${\verb|malloc|} : (n : {\verb|size_t|}) \to \top_n$
where $\top_n$ denotes the common supertype of all $n$-byte types.
The key feature is that the {\em value} of a parameter determines the {\em type} of the result.

Higher-rank 
types are needed to properly model functions that accept pointers to polymorphic functions
as parameters.  Such functions are not entirely uncommon; for example, any function that
is parameterized by a custom polymorphic allocator will have rank $\geq 2$.

Retypd was implemented as an inference phase that runs after CodeSurfer's main analysis loop.
We expect that by moving Retypd into CodeSurfer's analysis loop, there will be an opportunity
for interesting interactions between IR generation and type reconstruction.
 
\section {Conclusion}
By examining a diverse corpus of optimized binaries, we have identified a number of 
common idioms that are stumbling blocks for machine-code type inference.  For each
of these idioms, we identified a type-system feature that could enable the difficult code
to be properly typed.  We gathered these features into a type system and implemented 
the inference algorithm in the tool Retypd.  Despite removing the requirement for
points-to data, Retypd is able to accurately and conservatively
type a wide variety of real-world binaries.
%Our results with Retypd suggest that inference of polymorphic types from machine code is a promising direction for further study.
We assert that Retypd demonstrates the utility of high-level type systems for
reverse engineering and binary analysis.
 
\section*{Acknowledgments}
The authors would like to thank Vineeth Kashyap and the anonymous reviewers for their many useful comments on this manuscript, and John Phillips, David Ciarletta, and Tim Clark for their help with test automation.

\newpage
\medskip

\appendix

\section {Constraint Generation}
\label{sec-interp}

Type constraint generation is performed by a parameterized abstract interpretation
$\textsc{Type}_\textsc{A}$; the parameter $\textsc{A}$ is itself an abstract interpreter
that is used to transmit additional analysis information such as reaching definitions,
propagated constants, and value-sets (when available).

Let $\mathcal{V}$ denote the set of type variables and $\mathcal{C}$ the set of type constraints.
Then the primitive TSL value- and map- types for $\textsc{Type}_\textsc{A}$ are given by
\begin{align*}
\mathsf{BASETYPE}_{\textsc{Type}_\textsc{A}} &= \mathsf{BASETYPE}_\textsc{A} \times 2^\mathcal{V} \times 2^\mathcal{C}\\
\mathsf{MAP[}\alpha, \beta\mathsf{]}_{\textsc{Type}_\textsc{A}} &=
\mathsf{MAP[}\alpha, \beta\mathsf{]}_\textsc{A} \times 2^\mathcal{C}
\end{align*}

Since type constraint generation is a syntactic, flow-insensitive process, we can regain flow sensitivity
by pairing with an abstract semantics that carries a summary of flow-sensitive information.  Parameterizing
the type abstract interpretation by $\textsc{A}$ allows us to factor out the particular
way in which program variables should be abstracted to types (e.g. SSA form, reaching definitions, and so on).

\subsection{Register Loads and Stores}
The basic reinterpretations proceed by pairing with the abstract interpreter $\textsc{A}$. For example,

\begin{code}[mathescape]
regUpdate(s, reg, v) =
  let (v', t, c) = v
      (s', m)    = s
      s'' = regUpdate(s', reg, v')
      (u, c') = A(reg, s'')
  in
      ( s'', m $\cup$ c $\cup$ { t$\subtype$u } )
\end{code}

\noindent
where $\texttt{A(reg, s)}$ produces a type variable from the register $\texttt{reg}$ and the $\texttt{A}$-abstracted register map $\texttt{s''}$.

Register loads are handled similarly:

\begin{code}[mathescape]
regAccess(reg, s) =
  let (s', c) = s
      (t, c') = A(reg, s')
  in
      ( regAccess(reg, s'), t, c $\cup$ c' )
\end{code}

\begin{example}
    Suppose that $\textsc{A}$ represents the concrete semantics for x86 and
    $\texttt{A(reg, $\cdot$)}$ yields a type variable $(\mathsf{reg}, \{\})$ and no additional
    constraints.  Then the x86 expression
    {\verb|mov ebx, eax|} is represented by the TSL expression
    {\verb|regUpdate(S, EBX(), regAccess(EAX(), S))|}, where $\texttt{S}$ is the initial 
    state $(S_\text{conc}, \mathcal{C})$. After abstract interpretation, $\mathcal{C}$ will
    become $\mathcal{C} \cup \{ \mathsf{eax} \subtype \mathsf{ebx} \}$.
\end{example}

By changing the parametric interpreter $\textsc{A}$, the generated type constraints may be made
more precise.

\begin{example}
   We continue with the example of {\verb|mov ebx, eax|} above.
   Suppose that $\textsc{A}$ represents an abstract semantics that is aware of register reaching
   definitions, and define $\texttt{A(reg, s)}$ by

\begin{code}[mathescape]
   A(reg, s) =
    case reaching-defs(reg, s) of
      { p } $\rightarrow$ ($\mathsf{reg}_p$, {})
      defs $\rightarrow$
           let t = fresh
               c = { $\mathsf{reg}_p \subtype \mathsf{t}$ | p $\in$ defs }
           in (t, c)
\end{code}

\noindent
   where {\verb|reaching-defs|} yields the set of definitions of $\texttt{reg}$ that are
   visible from state $\texttt{s}$. Then $\textsc{Type}_\textsc{A}$ at program point $q$
   will update the constraint set $\mathcal{C}$ to
      \[\mathcal{C} \cup \{ \mathsf{eax}_p \subtype \mathsf{ebx}_q \}\]
   if $p$ is the lone reaching definition of {\verb|EAX|}. If there are multiple reaching
   definitions $P$, then the constraint set will become 
      \[\mathcal{C} \cup \{ \mathsf{eax}_p \subtype t ~|~ p \in P \} \cup \{ t \subtype \mathsf{ebx}_q \}\]
\end{example}

\subsection{Addition and Subtraction}
\label{add-and-sub}
It is useful to track translations of a value through additions or subtraction of a constant. To that end,
we overload the {\verb|add(x,y)|} and {\verb|sub(x,y)|} operations in the cases where
{\verb|x|} or {\verb|y|}
have statically-determined constant values.  For example, if {\verb|INT32(n)|} is a concrete numeric
value then

\begin{code}
add(v, INT32(n)) =
  let (v', t, c) = v in
    ( add(v', INT32(n)), t.+n, c )
\end{code}

In the case where neither operand is a statically-determined constant, we generate a fresh type variable
representing the result and a 3-place constraint on the type variables:

\begin{code}[mathescape]
add(x, y) =
  let (x', t$_1$, c$_1$) = x
      (y', t$_2$, c$_2$) = y
      t = $\text{fresh}$
  in
    ( add(x', y'),
      t,
      c$_1$ $\cup$ c$_2$ $\cup$ { $\textsc{Add}$(t$_1$, t$_2$, t) } )
\end{code}

\noindent
Similar interpretations are used for {\verb|sub(x,y)|}.

\subsection{Memory Loads and Stores}
Memory accesses are treated similarly to register accesses, except for the use of dereference accesses
and the handling of points-to sets.  For any abstract $\texttt{A}$-value $\texttt{a}$ and
$\texttt{A}$-state $\texttt{s}$, let {\verb|A(a,s)|} denote a set of type variables representing
the address $\texttt{A}$ in the context $\texttt{s}$. Furthermore, define $\texttt{PtsTo}_\textsc{A}(a,s)$
to be a set of type variables representing the values pointed to by $\texttt{a}$ in the context $\texttt{s}$.

The semantics of the $\texttt{N}$-bit load and store functions  $\texttt{memAccess$_N$}$ and
$\texttt{memUpdate$_N$}$ are given by

\begin{code}[mathescape]
memAccess$_{\texttt{N}}$(s, a) =
  let (s$_0$, c$_s$)     = s
      (a$_0$, t, c$_t$) = a
      c$_{pt}$ = { x$\subtype$t.load.$\sigma$N@0
           | x $\in$ PtsTo(a$_0$,s$_0$) }
  in
      ( memAccess$_{\texttt{N}}$(s$_0$, a$_0$),
        t.load.$\sigma$N@0,
        c$_s$ $\cup$ c$_t$ $\cup$ c$_{pt}$ )

memUpdate$_{\texttt{N}}$(s, a, v) =
  let (s$_0$, c$_s$)     = s
      (a$_0$, t, c$_t$)  = a
      (v$_0$, v, c$_v$) = v
      c$_{pt}$ = { t.store.$\sigma$N@0$\subtype$x
           | x $\in$ PtsTo(a$_0$,s$_0$) }
  in
      ( memUpdate$_{\texttt{N}}$(s$_0$, a$_0$, v$_0$),
        c$_s$ $\cup$ c$_t$ $\cup$ c$_v$ $\cup$ c$_{pt}$
          $\cup$ { v$\subtype$t.store.$\sigma$N@0 } )
\end{code}

We achieved acceptable results by using a bare minimum points-to analysis that only tracks
constant pointers to the local activation record or the data section. The use of the
$\mathsf{.load}$ / $\mathsf{.store}$ accessors allows us to track multi-level pointer 
information without the need for explicit points-to data.  The minimal approach tracks just enough
points-to information to resolve references to local and global variables.

\subsection{Procedure Invocation}
Earlier analysis phases are responsible for delineating procedures and gathering
data about each procedure's formal-in and formal-out variables, including 
information about how parameters are stored on the stack or in registers. This
data is transformed into a collection of {\em locators} associated to each
function.  Each locator is bound to a type variable representing the formal; the
locator is responsible for finding an appropriate set of type variables 
representing the actual at a callsite, or the corresponding local within the
procedure itself.

\begin{example} 
Consider this simple program that invokes a 32-bit identity function.

\begin{code}
 p:    push ebx  ; writes to local ext4
 q:    call id
        ...
id:    ; begin procedure id()
 r:    mov eax, [esp+arg0]
       ret
\end{code}

The procedure {\verb|id|} will have two locators:
\begin{itemize}
  \item A locator $L_i$ for the single parameter, bound to a type variable 
        $\mathsf{id}_{i}$.
  \item A locator $L_o$ for the single return value, bound to a type variable
        $\mathsf{id}_{o}$.
\end{itemize}
At the procedure call site, the locator $L_i$ will return the type variable
$\mathsf{ext4}_p$ representing the stack location {\verb|ext4|} tagged by
its reaching definition. Likewise, $L_o$ will return the type variable
$\mathsf{eax}_q$ to indicate that the actual-out is held in the version of
{\verb|eax|} that is defined at point $q$.  The locator results are combined
with the locator's type variables, resulting in the constraint set
\[\{ \mathsf{ext4}_p \subtype \mathsf{id}_i,\quad \mathsf{id}_o \subtype \mathsf{eax}_q \}\]

Within procedure {\verb|id|}, the locator $L_i$ returns the type variable
$\mathsf{arg0}_{\verb|id|}$ and $L_o$ returns $\mathsf{eax}_r$, resulting in the constraint set
\[\{ \mathsf{id}_i \subtype \mathsf{arg0}_{\verb|id|}, \quad \mathsf{eax}_r \subtype \mathsf{id}_o \}\]
 \end{example}

A procedure may also be associated with a set of type constraints between the locator type variables, 
called the {\em procedure summary};
these type constraints may be inserted at function calls to model the known behavior of a function.
For example, invocation of {\verb|fopen|} will result in the constraints
\[\{ \mathsf{fopen}_{i_0} \subtype \underline{\mathsf{char*}},\;
\mathsf{fopen}_{i_1} \subtype \underline{\mathsf{char*}},\;
\underline{\mathsf{FILE*}} \subtype \mathsf{fopen}_o \} \]

To support polymorphic function invocation, we instantiate fresh versions of the
 locator type variables that are
tagged with the current callsite; this prevents type variables from multiple invocations of the same
procedure from being linked.
\begin{example}[cont'd]
When using callsite tagging, the callsite constraints generated by the locators would be
\[\{\mathsf{ext4}_p \subtype \mathsf{id}_i^q, \quad \mathsf{id}_o^q \subtype \mathsf{eax}_q\}\]
\end{example}
The callsite tagging must also be applied to any procedure summary. For example, a call to
{\verb|malloc|} will result in the constraints 
\[\{ \mathsf{malloc}^p_i \subtype \underline{\mathsf{size\_t}},
\quad \underline{\mathsf{void*}} \subtype \mathsf{malloc}^p_o\}\]
If {\verb|malloc|} is used twice within a single procedure, we see an effect
like {\verb|let|}-polymorphism: each use will be typed independently.

\subsection{Other Operations}

\subsubsection{Floating-point}
Floating point types are produced by calls to known library functions and through an 
abstract interpretation of reads to and writes from the floating point register bank.
We do not track register-to-register moves between floating point registers, though it would
be straightforward to add this ability. In theory, this causes us to lose precision when attempting to
distinguish between typedefs of floating point values; in practice, such typedefs appear to be
extremely rare.

\subsubsection{Bit Manipulation}
We assume that the operands and results of most bit-manipulation operations are
integral, with some special exceptions:
\begin{itemize}
\item Common idioms like {\verb|xor reg,reg|} and {\verb|or reg,-1|} are used to
initialize registers to certain constants. On x86 these instructions can be encoded
with 8-bit immediates, saving space relative to the equivalent versions 
{\verb|mov reg,0|} and {\verb|mov reg,-1|}.  We do not assume that the results of these
operations are of integral type.
\item We discard any constraints generated while computing a value that is only used to
update a flag status. In particular, on x86 the operation {\verb|test reg1,reg2|}
is implemented like a bitwise-$\textsc{and}$ that discards its result, only retaining the effect
on the flags.
\item For specific operations such as $y := x \textsc{ and } \texttt{0xfffffffc}$
and $y := x \textsc { or } 1$, we act as if they were equivalent to $y := x$. This is
because these specific operations are often used for {\em bit-stealing}; for example,
 requiring pointers to be aligned on 4-byte boundaries frees the lower two bits of a
pointer for other purposes such as marking for garbage collection.
\end{itemize}
 
\subsection {Additive Constraints}
\label{additive-constraints}

The special constraints $\textsc{Add}$ and $\textsc{Sub}$ are used to conditionally
propagate information about which type variables represent pointers and which
represent integers when the variables are related through addition or subtraction.
The deduction rules for additive constraints are summarized in \autoref{add-sub-deduce}.
 We obtained good results
by inspecting the unification graph used for computing $\mathcal{L}(S_i)$; the
graph can be used to quickly determine whether
a variable has pointer- or integer-like capabilities.  In practice, the 
constraint set also should be updated with new subtype constraints as the
additive constraints are applied, and a fully applied constraint can be dropped
from $\mathcal{C}$.  We omit these details for simplicity.

\begin{figure*}
\centering
\newcommand{\isp}{$p$}
\newcommand{\isi}{$i$}
\newcommand{\mkp}{$P$}
\newcommand{\mki}{$I$}

\begin{tabular}{lcccccc|ccccccc}
\multicolumn{1}{c}{} & \multicolumn{6}{c}{$\textsc{Add}$} & \multicolumn{7}{c}{$\textsc{Sub}$} \\
\cline{2-14}\noalign{\smallskip}
$X$ & \isi & \mki & \isp & \mkp & \mki & \isi
    & \isi & \mki & \mkp & \mkp & \isp & \isp & \isp \\
$Y$ & \isi & \mki & \mki & \isi & \isp & \mkp
    & \mki & \isi & \isi & \isp & \mkp & \isi & \mki \\
$Z$ & \mki & \isi & \mkp & \isp & \mkp & \isp
    & \mki & \isi & \isp & \mki & \isi & \mkp & \isp \\ 

\end{tabular}
\caption{Inference rules for $\textsc{Add}(X,Y;Z)$ and $\textsc{Sub}(X,Y;Z)$.  Lower case
letters denote known integer or pointer types.  Upper case letters denote inferred types.
For example, the first column says that if $X$ and $Y$ are integral types in
an $\textsc{Add}(X,Y;Z)$ constraint, then $Z$ is integral as well.}
\label{add-sub-deduce}
\end{figure*}

\section{Normal Forms of Proofs}
\label{normal-form-appendix}

A finite constraint set with recursive constraints can have an infinite entailment closure. In order
to manipulate entailment closures efficiently, we need finite (and small) representations of
infinite constraint sets.  The first step towards achieving a finite representation of entailment
is to find a normal form for every derivation.  In \autoref{construction}, finite models of
entailment closure will be constructed that manipulate representations of these normal forms.

\begin{lemma}
For any statement $P$ provable from $\mathcal{C}$, there exists a derivation of $\mathcal{C} \proves P$
that does not use rules $\textsc{S-Refl}$ or $\textsc{T-Inherit}$.
\end{lemma}
\begin{proof}
The redundancy of $\textsc{S-Refl}$ is immediate. As for $\textsc{T-Inherit}$, any use of that rule in a
proof can be replaced by $\textsc{S-Field}_\oplus$ followed by $\textsc{T-Left}$, or $\textsc{S-Field}_\ominus$
followed by $\textsc{T-Right}$.
\end{proof}

In the following, we make the simplifying assumption that $\mathcal{C}$ is closed under $\textsc{T-Left}$,
$\textsc{T-Right}$, and $\textsc{T-Prefix}$. Concretely, we are requiring that
\begin{enumerate}
\item if $\mathcal{C}$ contains a subtype constraint $\alpha \subtype \beta$ as an axiom, it also contains
the axioms $\term{\alpha}$ and $\term{\beta}$.
\item if $\mathcal{C}$ contains a term declaration $\term{\alpha}$, it also contains term
declarations for all prefixes of $\alpha$.
\end{enumerate}

\begin{lemma}
If $\mathcal{C}$ is closed under $\textsc{T-Left}$, $\textsc{T-Right}$, and $\textsc{T-Prefix}$
then any statement provable from $\mathcal{C}$ can be proven without use of $\textsc{T-Prefix}$.
\end{lemma}
\begin{proof}
By the previous lemma, we may assume that $\textsc{S-Refl}$ and $\textsc{T-Inherit}$ are not used in
the proof.
We will prove the lemma by transforming all subproofs that end in a use of $\textsc{T-Prefix}$ to remove
that use.  To that end, assume we have a proof ending with the derivation of $\term{\alpha}$ from
$\term{\alpha.\ell}$ using $\textsc{T-Prefix}$.  We then enumerate the ways that $\term{\alpha.\ell}$
may have been proven.
\begin{itemize}
\item If $\term{\alpha.\ell} \in \mathcal{C}$ then $\term{\alpha} \in \mathcal{C}$ already, so the
entire proof tree leading to $\term{\alpha}$ may be replaced with an axiom from $\mathcal{C}$.
\item The only other way a $\term{\alpha.\ell}$ could be introduced is through $\textsc{T-Left}$ or
$\textsc{T-Right}$. For simplicity, let us only consider the $\textsc{T-Left}$ case; $\textsc{T-Right}$ is
similar.  At this point, we have a derivation tree that looks like
\[
\inferrule
  {\alpha.\ell \subtype \varphi}
  {\inferrule {\term{\alpha.\ell}} {\term{\alpha}}}
\]
How was $\alpha.\ell \subtype \varphi$ introduced? If it were an axiom of $\mathcal{C}$ then our assumed
closure property would imply already that $\term{\alpha} \in \mathcal{C}$. The other cases are:
\begin{itemize}
\item $\varphi = \beta.\ell$ and $\alpha.\ell \subtype \varphi$ was introduced through one of the
$\textsc{S-Field}$ axioms. In either case, $\alpha$ is the left- or right-hand side of one of the
antecedents to $\textsc{S-Field}$, so the whole subderivation including the $\textsc{T-Prefix}$ axiom
can be replaced with a single use of $\textsc{T-Left}$ or $\textsc{T-Right}$.
\item $\ell = \mathsf{.store}$ and $\varphi = \alpha.\mathsf{load}$, with $\alpha.\ell \subtype \varphi$
introduced via $\textsc{S-Pointer}$. In this case, $\term{\alpha.\ell}$ is already an antecedent to
$\textsc{S-Pointer}$, so we may elide the subproof up to $\textsc{T-Prefix}$ to inductively reduce
the problem to a simpler proof tree.
\item Finally, we must have $\alpha.\ell \subtype \varphi$ due to $\textsc{S-Trans}$. But in this case,
the left antecedent is of the form $\alpha.\ell \subtype \beta$; so once again, we may elide a subproof
to get a simpler proof tree.
\end{itemize}
Each of these cases either removes an instance of $\textsc{T-Prefix}$ or results in a strictly smaller
proof tree. It follows that iterated application of these simplification rules results in a proof of
$\term{\alpha}$ with no remaining instances of $\textsc{T-Prefix}$.
\end{itemize}
\end{proof}
\begin{corollary}
If $\mathcal{C} \proves \term{\alpha}$, then either $\term{\alpha} \in \mathcal{C}$, or
$\mathcal{C} \proves \alpha \subtype \beta, \beta \subtype \alpha$ for some $\beta$. 
\end{corollary}

To simplify the statement of our normal-form theorem, recall that we defined the
variadic rule $\textsc{S-Trans'}$ to remove the degrees of freedom from 
regrouping repeated application of $\textsc{S-Trans}$:
\[
\inferrule
  {\alpha_1 \subtype \alpha_2,\; \alpha_2 \subtype \alpha_3,\; \dots\; \alpha_{n-1} \subtype \alpha_n}
  {\alpha_1 \subtype \alpha_n}
  \; \textsc{(S-Trans')}
\]

\begin{theorem}[Normal form of proof trees]
\label{normal-form-thm}
Let $\mathcal{C}$ be a constraint set that is closed under $\textsc{T-Left}$, $\textsc{T-Right}$,
and $\textsc{T-Prefix}$.  Then any statement provable from $\mathcal{C}$ has a derivation such that
\begin{itemize}
\item There are no uses of the rules $\textsc{T-Prefix}$, $\textsc{T-Inherit}$, or $\textsc{S-Refl}$.
\item For every instance of $\textsc{S-Trans'}$ and every pair of adjacent antecedents, at least one
is {\em not} the result of a $\textsc{S-Field}$ application.
\end{itemize}
\end{theorem}
\begin{proof}
The previous lemmas already handled the first point. As to the second, suppose that we had two
adjacent antecedents that were the result of $\textsc{S-Field}_\oplus$:
\[
\inferrule
  {\cdots \inferrule {\inferrule{P} {\alpha \subtype \beta}\quad \inferrule {Q}{\term{\beta.\ell}}}{\alpha.\ell \subtype \beta.\ell} \quad \inferrule {\inferrule{R}{\beta \subtype \gamma}\quad \inferrule{S}{\term{\gamma.\ell}}}{\beta.\ell \subtype \gamma.\ell} \cdots}
  { T }
\]
First, note that the neighboring applications of $\textsc{S-Field}_\oplus$ must both use the same field label
$\ell$, or else $\textsc{S-Trans}$ would not be applicable. But now observe that we may move a $\textsc{S-Field}$
application upwards, combining the two $\textsc{S-Field}$ applications into one:
\[
\inferrule
  { \cdots
    \inferrule
      {\inferrule
        {\inferrule{P}{\alpha \subtype \beta} \quad \inferrule{R}{\beta \subtype \gamma}}
        {\alpha \subtype \gamma}
       \quad {\inferrule{S}{\term{\gamma.\ell}}}}
      {\alpha.\ell \subtype \gamma.\ell}
    \cdots}
  { T }
\]
For completeness, we also compute the simplifying transformation for adjacent uses of $\textsc{S-Field}_\ominus$:
the derivation
\[
\inferrule
  {\cdots \inferrule {\inferrule{P} {\beta \subtype \alpha}\quad \inferrule {Q}{\term{\alpha.\ell}}}{\alpha.\ell \subtype \beta.\ell} \quad \inferrule {\inferrule{R}{\gamma \subtype \beta}\quad \inferrule{S}{\term{\beta.\ell}}}{\beta.\ell \subtype \gamma.\ell} \cdots}
  { T }
\]
can be simplified to
\[
\inferrule
  { \cdots
    \inferrule
      {\inferrule{Q}{\term{\alpha.\ell}}
       \quad
       \inferrule
        {\inferrule{R}{\gamma \subtype \beta} \quad \inferrule{P}{\beta \subtype \alpha}}
        {\gamma \subtype \alpha}}
      {\alpha.\ell \subtype \gamma.\ell}
    \cdots}
  { T }
\]
We also can eliminate the need for the subderivations $Q$ and $S$ except for the very first and very last
antecedents of $\textsc{S-Trans'}$ by pulling the term out of the neighboring subtype relation; schematically,
we can always replace
\[
\inferrule
    {\cdots
       \inferrule{ \inferrule {P} {\alpha \subtype \beta}
                   \quad
                   \inferrule {S} {\term {\beta.\ell}} } 
                 { \alpha.\ell \subtype \beta.\ell }
       \quad
       \inferrule {U} {\beta.\ell \subtype \gamma}
     \cdots}
    {T}
\]
by
\[
\inferrule
    {\cdots
       \inferrule{ \inferrule {P} {\alpha \subtype \beta}
                   \quad
                   \inferrule { \inferrule{U}{\beta.\ell \subtype \gamma} } {\term {\beta.\ell}} } 
                 { \alpha.\ell \subtype \beta.\ell }
       \quad
       \inferrule {U} {\beta.\ell \subtype \gamma}
     \cdots}
    {T}
\]
This transformation even works if there are several $\textsc{S-Field}$ applications in a row, as in
\[
\inferrule
    {\cdots
       \inferrule { \inferrule{ \inferrule {P} {\alpha \subtype \beta}
                                  \quad
                                \inferrule {Q} {\term {\beta.\ell}}
                              } 
                              { \alpha.\ell \subtype \beta.\ell }
                   \quad { \inferrule {R} {\term {\beta.\ell.\ell'}} }
                  }
                  { \alpha.\ell.\ell' \subtype \beta.\ell.\ell' }
       \quad
       \inferrule {U} {\beta.\ell.\ell' \subtype \gamma}
     \cdots}
    {T}
\]
which can be simplified to
\[
\inferrule
    {\cdots
       \inferrule { \inferrule{ \inferrule {P} {\alpha \subtype \beta}
                                  \quad
                                \inferrule { \inferrule
                                              { \inferrule {U} {\beta.\ell.\ell' \subtype \gamma} } 
                                              {\term{\beta.\ell.\ell'}} } {\term {\beta.\ell}}
                              } 
                              { \alpha.\ell \subtype \beta.\ell }
                   \quad { \inferrule { \inferrule {U} {\beta.\ell.\ell' \subtype \gamma} }
                                      {\term {\beta.\ell.\ell'}} }
                  }
                  { \alpha.\ell.\ell' \subtype \beta.\ell.\ell' }
       \quad
       \inferrule {U} {\beta.\ell.\ell' \subtype \gamma}
     \cdots}
    {T}
\]
Taken together, this demonstrates that the $\term{\beta.\ell}$ antecedents to $\textsc{S-Field}$ are
automatically satisfied if the consequent is eventually used in an $\textsc{S-Trans}$ application
where $\beta.\ell$ is a ``middle'' variable. This remains true even if there are several $\textsc{S-Field}$
applications in sequence.
\end{proof}

Since the $\textsc{S-Field}$ $\textsc{Var}$ antecedents are automatically satisfiable, we can use a simplified
schematic depiction of proof trees that omits the intermediate $\textsc{Var}$ subtrees.  In this simplified
depiction, the previous proof tree fragment would be written as
\[
\inferrule
    {\cdots
      \inferrule {
        \inferrule {
          \inferrule { P }
                     { \alpha \subtype \beta } }
        {\alpha.\ell \subtype \beta.\ell} }
      {\alpha.\ell.\ell' \subtype \beta.\ell.\ell'}
      \quad
      \inferrule {U} {\beta.\ell.\ell' \subtype \gamma}
      \cdots}
    {T}
\]

Any such fragment can be automatically converted to a full proof by re-generating $\textsc{Var}$
subderivations from the results of the input derivations $P$ and/or $U$.

\subsection{Algebraic Representation}
\label{alg-rep}
Finally, consider the form of the simplified proof tree that elides the $\textsc{Var}$ subderivations.
Each leaf of the tree is a subtype constraint $c_i \in \mathcal{C}$, and there is a unique way to
fill in necessary $\textsc{Var}$ antecedents and $\textsc{S-Field}$ applications to glue the
constraints $\{\mathcal{c}_i\}$ into a proof tree in normal form. In other words, the normal
form proof tree is completely determined by the sequence of leaf constraints $\{c_i\}$ and
the sequence of $\textsc{S-Field}$ applications applied to the final tree.  In effect, we have
a term algebra with constants $C_i$ representing the constraints $c_i \in \mathcal{C}$, an
associative binary operator $\odot$ that combines two compatible constraints via $\textsc{S-Trans}$,
inserting appropriate $\textsc{S-Field}$ applications, and a unary operator $S_\ell$ for each
$\ell \in \Sigma$ that represents an application of $\textsc{S-Field}$. The proof tree manipulations
provide additional relations on this term algebra, such as
\[S_\ell(R_1) \odot S_\ell(R_2) = \begin{cases}
S\ell(R_1 \odot R_2) & \text{when } \langle \ell \rangle = \oplus \\
S\ell(R_2 \odot R_1) & \text{when } \langle \ell \rangle = \ominus \\
\end{cases}\]
The normalization rules described in this section demonstrate that every proof of a subtype
constraint entailed by $\mathcal{C}$ can be represented by a term of the form
\[S_{\ell_1}\left(S_{\ell_2}\left(\cdots S_{\ell_n}(R_1 \odot \cdots \odot R_k)\cdots \right)\right)\]

\section{The $\mathsf{StackOp}$ Weight Domain}

In this section, we develop a weight domain $\mathsf{StackOp}_\Sigma$ that will be useful for
modeling constraint sets by pushdown systems. This weight domain is generated by symbolic 
constants representing actions on a stack. 

\begin{definition}
The weight domain $\mathsf{StackOp}_\Sigma$ is the idempotent $*$-semiring generated by the symbols
$\stPop{x}$, $\stPush{x}$ for all $x \in \Sigma$, subject only to the relation
\[\stPush{x} \otimes \stPop{y} = \delta(x,y) =  \begin{cases} \one & \text{if}~ x = y \\ \zero & \text{if} ~ x \neq y \end{cases}\]
To simplify the presentation, when $u = u_1 \cdots u_n$ we will sometimes write
\[\stPop{u} = \stPop{u_1} \otimes \cdots \otimes \stPop{u_n}\]
and
\[\stPush{u} = \stPush{u_n} \otimes \cdots \otimes \stPush{u_1}\]
\end{definition}

\begin{definition}
A monomial in $\mathsf{StackOp}_\Sigma$ is called {\em reduced} if its length cannot be shortened through applications of the $\stPush{x} \otimes \stPop{y} = \delta(x,y)$ rule. Every reduced monomial has the form
\[\stPop{u_1} \cdots \stPop{u_n} \otimes \stPush{v_m} \cdots \stPush{v_1}\]
for some $u_i, v_j \in \Sigma$. 
\end{definition}

Elements of $\mathsf{StackOp}_\Sigma$ can be understood to denote possibly-failing functions operating on
a stack of symbols from $\Sigma$:
\begin{align*}
\denotes{\zero} &= \lambda s ~.~ \mathsf{fail} \\
\denotes{\one} &= \lambda s ~.~ s \\
\denotes{\stPush{x}} &= \lambda s ~.~ \mathsf{cons}(x,s)\\
\denotes{\stPop{x}} &= \lambda s ~.~ \text{case} ~ s ~ \text{of} \\
& \qquad \mathsf{nil} \to \mathsf{fail} \\
& \qquad \mathsf{cons}(x',s') \to \text{if }x = x'\text{ then }s' \\
& \hspace{.262\textwidth} \text{else }\mathsf{fail}
\end{align*}
Under this interpretation, the semiring operations are interpreted as 
\begin{align*}
\denotes{X \oplus Y} &= \text{nondeterministic choice of } \denotes{X} ~ \text{or} ~ \denotes{Y} \\
\denotes{X \otimes Y} &= \denotes{Y} \circ \denotes{X} \\
\denotes{X^*} &= \text{nondeterministic iteration of } \denotes{X}
\end{align*}

Note that the action of a pushdown rule $R = \pdsrule{P}{u}{Q}{v}$ on a stack configuration 
$c = (X,w)$ is given by the application
\[\denotes{\stPop{P} \otimes \stPop{u} \otimes \stPush{v} \otimes \stPush{Q}} \left( \mathsf{cons}(X,w) \right)\]
The result will either be $\mathsf{cons}(X',w')$ where $(X',w')$ is the configuration obtained by applying rule $R$ to configuration $c$, or $\mathsf{fail}$ if the rule cannot be applied to $c$.

\begin{lemma}
There is a one-to-one correspondence between elements of $\mathsf{StackOp}_\Sigma$ and regular sets
of pushdown system rules over $\Sigma$.
\end{lemma}

\begin{definition}
The variance operator $\langle \cdot \rangle$ can be extended to $\mathsf{StackOp}_\Sigma$ by defining \[\langle \stPop{x} \rangle = \langle \stPush{x} \rangle = \langle x \rangle\]
\end{definition}
 
\section{Constructing the Transducer for a Constraint Set}
\label{construction}

\subsection{Building the Pushdown System}

Let $\mathcal{V}$ be a set of type variables and $\Sigma$ a set of field labels, equipped with
a variance operator $\langle \cdot \rangle$.  Furthermore, suppose that we have fixed a set
$\mathcal{C}$ of constraints over $\mathcal{V}$.  Finally, suppose that we can partition $\mathcal{V}$
into a set of {\em interesting} and {\em uninteresting} variables:
\[\mathcal{V} = \mathcal{V}_i \amalg \mathcal{V}_u\]

\begin{definition}
Suppose that $X, Y \in \mathcal{V}_i$ are interesting type variables, and there is a proof
$\mathcal C \proves X.u \subtype Y.v$.  We will call the proof {\em elementary} if the normal
form of its proof tree only involves uninteresting variables on internal leaves. We write 
\[\mathcal{C} \proves^{\mathcal{V}_i}_\text{elem} X.u \subtype Y.v\]
when $\mathcal{C}$ has an elementary proof of $X.u \subtype Y.v$.
\end{definition}

Our goal is to construct a finite state transducer that recognizes the relation
$\subtype^{\mathcal{V}_i}_\text{elem} \subseteq \mathcal{V}'_i \times \mathcal{V}'_i$ between
derived type variables defined by
\[\subtype^{\mathcal{V}_i}_\text{elem} = \left\{ (X.u, Y.v) ~|~ \mathcal{C}
\proves^{\mathcal{V}_i}_\text{elem} X.u \subtype Y.v\right\}\]
We proceed by constructing an unconstrained pushdown system $\mathcal{P}_\mathcal{C}$ whose derivations
model proofs in $\mathcal{C}$; a modification of Caucal's saturation algorithm \cite{caucal1992regular} is then used
to build the transducer representing $\Deriv{P_C}$.

\begin{definition}
Define the left- and right-hand tagging rules $\mathsf{lhs}$, $\mathsf{rhs}$ by
\[\mathsf{lhs}(x) = \begin{cases} x_{\mathsf{L}} & \text{if} ~ x \in \mathcal{V}_i \\
x &  \text{if} ~ x \in \mathcal{V}_o \end{cases}\]
\[\mathsf{rhs}(x) = \begin{cases} x_{\mathsf{R}} & \text{if} ~ x \in \mathcal{V}_i \\
x & \text{if} ~ x \in \mathcal{V}_o \end{cases}\]
\end{definition}

\begin{definition}
The unconstrained pushdown system $\mathcal{P}_\mathcal{C}$ associated to a constraint set $\mathcal{C}$
is given by the triple $(\widetilde{\mathcal{V}}, \widetilde{\Sigma}, \Delta)$ where
\begin{align*}
\widetilde{\mathcal{V}} &= \left(\mathsf{lhs}(\mathcal{V}_i) \amalg \mathsf{rhs}(\mathcal{V}_i) \amalg \mathcal{V}_o\right) \times \{ \oplus, \ominus \} \\
& \cup \{ \stStart, \stEnd \}
\end{align*}
We will write an element $(v, \odot) \in \widetilde{\mathcal{V}}$ as $v^\odot$.

The stack alphabet $\widetilde{\Sigma}$ is essentially the same as $\Sigma$, with a few extra tokens added
to represent interesting variables:
\[\widetilde{\Sigma} = \Sigma \cup \{ v^\oplus ~|~ v \in \mathcal{V}_i \} \cup \{ v^\ominus ~|~ v \in \mathcal{V}_i \}\]
To define the transition rules, we first introduce the helper functions:
\begin{align*}
\mathsf{rule}^\oplus(p.u \subtype q.v) &= \pdsrule{\mathsf{lhs}(p)^{\langle u \rangle}}{u}{\mathsf{rhs}(q)^{\langle v \rangle}}{v} \\
\mathsf{rule}^\ominus(p.u \subtype q.v) &= \pdsrule{\mathsf{lhs}(q)^{\ominus \cdot \langle v \rangle}}{v}{\mathsf{rhs}(p)^{\ominus \cdot \langle u \rangle}}{u} \\
\mathsf{rules}(c) &= \{ \mathsf{rule}^\oplus(c),~ \mathsf{rule}^\ominus(c)\}
\end{align*}
The transition rules $\Delta$ are partitioned into four parts 
$\Delta = \Delta_\mathcal{C} \amalg \Delta_\mathsf{ptr} \amalg \Delta_\mathsf{start} \amalg \Delta_\mathsf{end}$ where
\begin{align*}
\Delta_\mathcal{C} &= \bigcup_{c \in \mathcal{C}} \mathsf{rules}(c) \\
\Delta_\mathsf{ptr} &= \bigcup_{v \in \mathcal{V}'} \mathsf{rules}(v.\mathsf{store} \subtype v.\mathsf{load}) \\
\Delta_\mathsf{start} &= \left\{\pdsrule{\stStart}{v^\oplus}{v^\oplus_\mathsf{L}}{\varepsilon}~|~v \in \mathcal{V}_i \right\}\\
 &\cup \left\{\pdsrule{\stStart}{v^\ominus}{v^\ominus_\mathsf{L}}{\varepsilon}~|~v \in \mathcal{V}_i \right\}\\
\Delta_\mathsf{end} &= \left\{\pdsrule{v^\oplus_\mathsf{R}}{\varepsilon}{\stEnd}{v^\oplus}~|~v \in \mathcal{V}_i \right\}\\
 &\cup \left\{\pdsrule{v^\ominus_\mathsf{R}}{\varepsilon}{\stEnd}{v^\ominus}~|~v \in \mathcal{V}_i \right\}
\end{align*}
\end{definition}

\begin{note}
The $\{ \oplus, \ominus \}$ superscripts on the control states are used to track the current variance of
the stack state, allowing us to distinguish between uses of an axiom in co- and contra-variant position.
The tagging operations $\mathsf{lhs}$, $\mathsf{rhs}$ are used to prevent derivations from making use
of variables from $\mathcal{V}_i$, preventing $\mathcal{P}_\mathcal{C}$ from admitting  derivations that
represent non-elementary proofs.
\end{note}

\begin{note}
Note that although $\Delta_\mathcal{C}$ is finite,  $\Delta_\mathsf{ptr}$ contains rules for {\em every
derived type variable} and is therefore infinite. We carefully adjust for this in the saturation
rules below so that rules from $\Delta_\mathsf{ptr}$ are only considered lazily as $\Deriv{P_C}$ is
constructed.
\end{note}

\begin{lemma}
For any pair $(X^a u, Y^b v)$ in $\Deriv{P_C}$, we also have
$(Y^{\ominus \cdot b} v, X^{\ominus \cdot a} u) \in \Deriv{P_C}$.
\end{lemma}
\begin{proof}
Immediate due to the symmetries in the construction of $\Delta$.
\end{proof}

\begin{lemma}
For any pair $(X^a u, Y^b v)$ in $\Deriv{P_C}$, the relation
\[a \cdot \langle u \rangle = b \cdot \langle v \rangle\]
must hold.
\label{deriv-signs}
\end{lemma}
\begin{proof}
This is an inductive consequence of how $\Delta$ is defined: the relation holds for every
rule in $\Delta_\mathcal{C} \cup \Delta_\mathsf{ptr}$ due to the use of the $\mathsf{rules}$ function,
and the sign in the exponent is propagated from the top of the stack during an application
of a rule from $\Delta_\mathsf{start}$ and back to the stack during an application of a rule 
from $\Delta_\mathsf{end}$. Since every derivation will begin with a rule from $\Delta_\mathsf{start}$,
proceed by applying rules from $\Delta_\mathcal{C} \cup \Delta_\mathsf{ptr}$, and conclude with
a rule from $\Delta_\mathsf{end}$, this completes the proof.
\end{proof}

\begin{definition}
Let $\Deriv{P_C}' \subseteq \mathcal{V}'_i \times \mathcal{V}'_i$ be the relation on derived type
variables induced by $\Deriv{P_C}$ as follows:
\[\Deriv{P_C}' = \left\{ (X.u, Y.v) ~|~ (X^{\langle u \rangle}u, Y^{\langle v \rangle} v)\in \Deriv{P_C}\right\}\]
\end{definition}

\begin{lemma}
If $(p,q) \in \Deriv{P_C}'$ and $\sigma \in \Sigma$ has $\langle \sigma \rangle = \oplus$, then
$(p \sigma, q \sigma) \in \Deriv{P_C}'$ as well.  If $\langle \sigma \rangle = \ominus$, then $(q \sigma, p \sigma) \in \Deriv{P_C}'$ instead.
\label{deriv-suffix}
\end{lemma}
\begin{proof}
Immediate, from symmetry considerations.
\end{proof}

\begin{lemma}
$\Deriv{P_C}$ can be partitioned into $\Deriv{P_C}^\oplus \amalg \Deriv{P_C}^\ominus$ where
\begin{align*}
\Deriv{P_C}^\oplus &= \left\{ (X^{\langle u \rangle} u, Y^{\langle v \rangle} v) ~|~ (X.u,Y.v) \in \Deriv{P_C}' \right\} \\
\Deriv{P_C}^\ominus &= \left\{ (Y^{\ominus \cdot \langle v \rangle} v, X^{\ominus \cdot \langle u \rangle} u) ~|~ (X.u,Y.v) \in \Deriv{P_C}' \right\} \\
\end{align*}
In particular, $\Deriv{P_C}$ can be entirely reconstructed from the simpler set $\Deriv{P_C}'$.
\label{deriv-parts}
\end{lemma}

\begin{theorem}
$\mathcal{C}$ has an elementary proof of the constraint $X.u \subtype Y.v$ if and only if
$\left(X . u, Y . v\right) \in \Deriv{P_C}'$.
\end{theorem}

\begin{proof}

We will prove the theorem by constructing a bijection between elementary proof trees and derivations
in $\mathcal{P}_\mathcal{C}$.

Call a configuration $(p^a, u)$ {\em positive} if $a = \langle u \rangle$. By \autoref{deriv-signs},
if $(p^a, u)$ is positive and $(p^a, u) \rewrites (q^b, v)$ then $(q^b, v)$ is positive as well.  We
will also call a transition rule positive when the configurations appearing in the rule are positive.
Note that by construction, every positive transition rule from $\Delta_\mathcal{C} \cup \Delta_\mathsf{ptr}$ is
of the form $\mathsf{rule}^\oplus(c)$ for some axiom $c \in \mathcal{C}$ or load/store constraint
$c = (p.\mathsf{store} \subtype p.\mathsf{load})$.

For any positive transition rule $R = \mathsf{rule}^\oplus(p.u \subtype q.v)$ and $\ell \in \Sigma$,
let $S_\ell(R)$ denote the positive rule
\[S_\ell(R) = \begin{cases}
\pdsrule{p^{\langle u \ell \rangle}}{u \ell}{q^{\langle v \ell \rangle}}{v \ell} & \text{if } \langle \ell \rangle = \oplus \\
\pdsrule{q^{\langle v \ell\rangle}}{v \ell}{p^{\langle u \ell \rangle}}{u \ell} & \text{if } \langle \ell \rangle = \ominus
\end{cases}\]

Note that by the construction of $\Delta$, each $S_\ell(R)$ is redundant, being a restriction of one of the existing
transition rules $\mathsf{rule}^\oplus(p.u \subtype q.v)$ or $\mathsf{rule}^\ominus(p.u \subtype q.v)$ to a more
specific stack configuration.  In particular, adding $S_\ell(R)$ to the set of
rules does not affect $\Deriv{P_C}$.

Suppose we are in the positive state $(p^{\langle w \rangle}, w)$ and are about to apply a transition rule
$R = \pdsrule{p^{\langle u \rangle}}{u}{q^{\langle v \rangle}}{v}$. Then we must have $w = u\ell_1 \ell_2 \cdots \ell_n$, so
that the left-hand side of the rule $S_{\ell_n}(\cdots(S_{\ell_1}(R)) \cdots )$ is exactly
$\langle p^{\langle w \rangle}; w \rangle$.

Let $R_\mathsf{start}, R_1, \cdots, R_n, R_\mathsf{end}$ be the sequence of rule applications used in an arbitrary
derivation $(\stStart, p^{\langle u \rangle} u) \rewrites (\stEnd, q^{\langle v \rangle} v)$ and let $R \odot R'$ denote
the application of rule $R$ followed by rule $R'$, with the right-hand side of $R$ {\em exactly} matching the
left-hand side of $R'$.
Given the initial stack state
$p^{\langle u \rangle} u$ and the sequence of rule applications, for each $R_k$ there is a unique rule 
$R_k' = S_{\ell^k_n}( \cdots (S_{\ell^k_1}(R_k))\cdots)$ such that
 $\partial = R'_\mathsf{start} \odot R'_1 \odot \cdots \odot R'_n \odot R'_\mathsf{end}$ describes the 
derivation exactly. Now normalize $\partial$ using the rule  $S_\ell(R_i) \odot S_\ell(R_j) \mapsto S_\ell(R_i \odot R_j)$;
this process is clearly reversible, producing a bijection between these normalized expressions and derivations
in $\Deriv{P_C}^\oplus$.

To complete the proof, we obtain a bijection between normalized expressions and the normal forms of elementary
proof trees. The bijection is straightforward: $R_i$ represents the introduction of an axiom from $\mathcal{C}$,
$\odot$ represents an application of the rule $\textsc{S-Trans}$, and $S_\ell$ represents an application of
the rule $\textsc{S-Field}_\ell$. This results in an algebraic expression describing the normalized
proof tree as in \autoref{alg-rep}, which completes the proof.
\end{proof}
 
\subsection{Constructing the Initial Graph}

In this section, we will construct a finite state automaton that accepts strings that encode some of the
behavior of $\mathcal{P_C}$; the next section describes a saturation algorithm that modifies the initial
automaton to create a finite state transducer for $\Deriv{P_C}'$.

The automaton $\mathcal{A_C}$ is constructed as follows:
\begin{enumerate}

\item Add an initial state $\stStart$ and an accepting state $\stEnd$.

\item For each left-hand side $\langle p^a ; u_1 \cdots u_n \rangle$ of a rule in $\Delta_\mathcal{C}$,
add the transitions
\begin{align*}
\stStart &\stackrel{\stPop{p^{a \cdot \langle u_1 \dots u_n \rangle}}}{\rewrite} p^{a \cdot \langle u_1 \dots u_n \rangle} \\
p^{a \cdot \langle u_1 \dots u_n \rangle} &\stackrel{\stPop{u_1}}{\rewrite} p^{a \cdot \langle u_2 \dots u_n \rangle} u_1 \\
&\vdots \\
p^{a \cdot \langle u_n \rangle} u_1 \dots u_{n-1} &\stackrel{\stPop{u_n}}{\rewrite} p^a u_1 \dots u_n
\end{align*}

\item For each right-hand side $\langle q^b ; v_1 \cdots v_n \rangle$ of a rule in $\Delta_\mathcal{C}$,
add the transitions
\begin{align*}
q^b v_1 \dots v_n &\stackrel{\stPush{v_n}}{\rewrite} q^{b \cdot \langle v_n \rangle} v_1 \dots v_{n-1} \\
&\vdots \\
q^{b \cdot \langle v_2 \dots v_n \rangle} v_1 &\stackrel{\stPush{v_1}}{\rewrite} q^{b \cdot \langle v_1 \dots v_n \rangle} \\
 q^{b \cdot \langle v_1 \dots v_n \rangle} &\stackrel{\stPush{ q^{b \cdot \langle v_1 \dots v_n \rangle}}}{\rewrite} \stEnd
\end{align*}

\item For each rule $\pdsrule{p^a}{u}{q^b}{v}$, add the transition
\[p^a u \stackrel{\one}{\rewrite} q^b v\]

\end{enumerate}

\begin{algorithm*}
\caption{Converting a transducer from a set of pushdown system rules}
\label{initial-automaton-code}
\begin{algorithmic}[1]
  
  \Procedure{AllPaths}{$V,E,x,y$}
  \Comment{\parbox[t]{.6\linewidth}{Tarjan's path algorithm. Return a finite state automaton recognizing
the label sequence for all paths from $x$ to $y$ in the graph $(V,E)$.}}
  \State $\dots$
  \State\Return $Q$
  \EndProcedure

  \Statex

  \Procedure{Transducer}{$\Delta_\mathcal{C}$}
  \Comment{\parbox[t]{.45\linewidth}{$\Delta_\mathcal{C}$ is a set of pushdown system rules $\pdsrule{p^a}{u}{q^b}{v}$}}
  \State{$V \gets \{\stStart, \stEnd\}$}
  \State{$E \gets \emptyset$}
  \ForAll{$\pdsrule{p^a}{u_1 \cdots u_n}{q^b}{v_1 \cdots v_m} \in \Delta_\mathcal{C}$}
    \State{$V \gets V \cup \{p^{a \cdot \langle u_1 \cdots u_n\rangle},q^{b \cdot \langle v_1 \cdots v_m\rangle}\}$}
    \State{$E \gets E \cup \{ (\stStart, p^{a \cdot \langle u_1 \cdots u_n \rangle}, \stPop{p^{a \cdot \langle u_1 \cdots u_n \rangle}}) \}$}
    \State{$E \gets E \cup \{ (q^{b \cdot \langle v_1 \cdots v_m \rangle}, \stEnd, \stPush{q^{b \cdot \langle v_1 \cdots v_m \rangle}}) \}$}
    \For{$i \gets 1 \dots n-1$}
      \State{$V \gets V \cup \{p^{a \cdot \langle u_{i+1} \cdots u_n\rangle}u_1 \dots u_i\}$}
      \State{$E \gets E \cup \{(p^{a \cdot \langle u_i \cdots u_n\rangle}u_1 \dots u_{i-1}, p^{a \cdot \langle u_{i+1} \cdots u_n \rangle}u_1 \dots u_i, 
      \stPop{u_i})\}$}
    \EndFor
    \For{$j \gets 1 \dots m-1$}
      \State{$V \gets V \cup \{q^{b \cdot \langle v_{j+1} \cdots u_m \rangle}v_1 \dots v_j\}$}
      \State{$E \gets E \cup \{(q^{b \cdot \langle v_{j+1} \cdots v_m \rangle}v_1 \dots v_j, q^{b \cdot \langle v_j \cdots v_m\rangle}v_1 \dots v_{j-1}, \stPush{v_j})\}$}
    \EndFor
    \State{$E \gets E \cup \{(p^a u_1 \dots u_n, q^b v_1 \dots v_m, \one)\}$}
  \EndFor
  \State{$E \gets$ \Call{Saturated}{$V,E$}}
  \State{$Q \gets$ \Call{AllPaths}{$V,E,\stStart,\stEnd$}}
  \State\Return $Q$
  \EndProcedure

\end{algorithmic}
\end{algorithm*}

\begin{definition}
Following \citet{carayol}, we call a sequence of transitions $\mathsf{p}_1 \cdots \mathsf{p}_n$ {\em productive}
if the corresponding element $\mathsf{p}_1 \otimes \cdots \otimes \mathsf{p}_n$ in $\mathsf{StackOp}_\Sigma$ is
not $\zero$. A sequence is productive if and only if it contains no adjacent terms of the form
$\stPop{y}, \stPush{x}$ with $x \neq y$.
\end{definition}

\begin{lemma}
There is a one-to-one correspondence between productive transition sequences accepted by $\mathcal{A_C}$ and
elementary proofs in $\mathcal{C}$ that do not use the axiom $\textsc{S-Pointer}$.
\label{sat1}
\end{lemma}
\begin{proof}
A productive transition sequence must consist of a sequence of pop edges followed by a sequence of
push edges, possibly with insertions of unit edges or intermediate sequences of the form
\[\stPush(u_n) \otimes \cdots \otimes \stPush(u_1) \otimes \stPop(u_1) \otimes \cdots \otimes \stPop(u_n)\]
Following a push or pop edge corresponds to observing or forgetting part of the stack. Following a 
$\one$ edge corresponds to applying a PDS rule from $\mathcal{P_C}$.  
\end{proof}
 
\subsection{Saturation}

The label sequences appearing in \autoref{sat1} are tantalizingly close to having the simple structure
of building up a pop sequence representing an initial state of the pushdown automaton, then building up
a push sequence representing a final state.  But the intermediate terms of the form ``push $u$, then pop it
again'' are unwieldy.  To remove the necessity for those sequences, we can {\em saturate} $\mathcal{A}_\mathcal{C}$
by adding additional $\one$-labeled transitions providing shortcuts to the push/pop subsequences.
We modify the standard saturation algorithm to also lazily instantiate transitions which correspond to
uses of the $\textsc{S-Pointer}$-derived rules in $\Delta_\mathsf{ptr}$.

\begin{lemma}
The reachable states of the automaton $\mathcal{A_C}$ can be partitioned into {\em covariant} and {\em contravariant} states, where a state's variance is defined to be the
variance of any sequence reaching the state from $\stStart$.
\label{state-variance}
\end{lemma}
\begin{proof}
By construction of $\Delta_\mathcal{C}$ and $\mathcal{A}_\mathcal{C}$.
\end{proof}

\begin{lemma}
There is an involution $n \mapsto \overline{n}$ on $\mathcal{A_C}$ defined by
\[\overline{x^a u} = x^{\ominus \cdot a} u\]
\[\overline{\stStart} = \stEnd\]
\[\overline{\stEnd} = \stStart\]
\label{path-symmetry}
\end{lemma}
\begin{proof}
Immediate due to the use of the rule constructors $\mathsf{rule}^\oplus$ and $\mathsf{rule}^\ominus$ when
forming $\Delta_\mathcal{C}$.
\end{proof}

In this section, the automaton $\mathcal{A_C}$ will be {\em saturated} by adding transitions to
create a new automaton $\mathcal{A}_\mathcal{C}^\mathsf{sat}$.

\begin{definition}
A sequence is called {\em reduced} if it is productive and contains no factors of the form
$\stPop{x} \otimes \stPush{x}$.
\end{definition}
Reduced productive sequences all have the form of a sequence of pops, followed by a sequence of pushes.

\noindent
The goal of the saturation algorithm is twofold:
\begin{enumerate}
\item Ensure that for any productive sequence accepted by $\mathcal{A_C}$ there is an equivalent
reduced sequence accepted by $\mathcal{A}_\mathcal{C}^\mathsf{sat}$.
\item Ensure that $\mathcal{A}_\mathcal{C}^\mathsf{sat}$ can represent elementary proofs that use
$\textsc{S-Pointer}$.
\end{enumerate}
The saturation algorithm \ref{sat-alg-code} proceeds by maintaining, for each state $q \in \mathcal{A_C}$, a set of
reaching-pushes $R(q)$.  The reaching push set $R(q)$ will contain the pair $(\ell, p)$ only
if there is a transition sequence in $\mathcal{A_C}$ from $p$ to $q$ with weight $\stPush{\ell}$.
When $q$ has an outgoing $\stPop{\ell}$ edge to $q'$ and $(\ell, p) \in R(q)$, we add a new
transition $p \stackrel{\one}{\to} q'$ to $\mathcal{A_C}$.

A special propagation clause is responsible for propagating reaching-push facts as if rules from
$\Delta_\mathsf{ptr}$ were instantiated, allowing the saturation algorithm to work even though the
corresponding unconstrained pushdown system has infinitely many rules.
This special clause is justified by considering
the standard saturation rule when $x.\mathsf{store} \subtype x.\mathsf{load}$ is
added as an axiom in $\mathcal{C}$. An example appears in \autoref{pointer-sat}:
a saturation edge is added from $x^\oplus.\mathsf{store}$ to $y^\oplus.\mathsf{load}$
due to the pointer saturation rule, but the same edge would also have been added if the 
states and transitions corresponding to $p.\mathsf{store} \subtype p.\mathsf{load}$
(depicted with dotted edges and nodes) were added to $\mathcal{A}_\mathcal{C}$.

\begin{algorithm*}
\caption{Saturation algorithm}
\label{sat-alg-code}
\begin{algorithmic}[1]

  \Procedure{Saturated}{$V,E$}
     \Comment{$V$ is a set of vertices partitioned into $V = V^\oplus \amalg V^\ominus$}
     \State{$E' \gets E$}
     \Comment{$E$ is a set of edges, represented as triples $(\mathsf{src}, \mathsf{tgt}, \mathsf{label})$}
     \ForAll{$x \in V$}
        \State{$R(x) \gets \emptyset$}
     \EndFor
     \ForAll{$(x,y,e) \in E$ with $e = \stPush{\ell}$}
        \Comment{Initialize the reaching-push sets $R(x)$}
        \State{$R(y) \gets R(y) \cup \{(\ell,x)\}$}
     \EndFor
     \Repeat
       \State{$R_\text{old} \gets R$}
       \State{$E'_\text{old} \gets E'$}
       \ForAll{$(x,y,e) \in E'$ with $e = \one$}
         \State{$R(y) \gets R(y) \cup R(x)$}
       \EndFor
       \ForAll{$(x,y,e) \in E'$ with $e = \stPop{\ell}$}
         \ForAll{$(\ell,z) \in R(x)$}
           \Comment{The standard saturation rule.}
           \State{$E' \gets E' \cup \{(z,y,\one)\}$}
         \EndFor
       \EndFor
       \ForAll{$x \in V^\ominus$}
         \Comment{Lazily apply saturation rules corresponding to $\textsc{S-Pointer}$.}
         \ForAll{$(\ell, z) \in R(x)$ with $\ell = \mathsf{.store}$}
           \Comment{See \autoref{pointer-sat} for an example.}
           \State{$R(\overline{x}) \gets R(\overline{x}) \cup \{(\mathsf{.load},z)\}$}
         \EndFor
         \ForAll{$(\ell, z) \in R(x)$ with $\ell = \mathsf{.load}$}
           \State{$R(\overline{x}) \gets R(\overline{x}) \cup \{(\mathsf{.store},z)\}$}
         \EndFor
       \EndFor
     \Until{$R = R_\text{old}$ and $E' = E'_\text{old}$}
     \State\Return{$E'$}
  \EndProcedure

\end{algorithmic}
\end{algorithm*}

\begin{figure*}
\begin{minipage}{0.7\textwidth}
\[\begin{tikzpicture}[scale=0.7]

\node [draw] at (0,10) (s) {$\textsc{\#Start}$};
\node [draw] at (0,8) (a) {$A_\mathsf{L}^\oplus$};
\node [draw] at (0,6) (xs) {$x.\mathsf{store}^\oplus$};
\node [draw] at (3,7) (x) {$x^\ominus$};
\node [draw] at (6,8) (pm) {$p^\ominus$};
\node [draw,style=dotted] at (5,6) (psp) {$p.\mathsf{store}^\oplus$};
\node [draw,style=dotted] at (7,6) (plm) {$p.\mathsf{load}^\ominus$};
\node [draw,style=dotted] at (5,4) (plp) {$p.\mathsf{load}^\oplus$};
\node [draw,style=dotted] at (7,4) (psm) {$p.\mathsf{store}^\ominus$};
\node [draw] at (6,2) (pp) {$p^\oplus$};
\node [draw] at (3,3) (y) {$y^\oplus$};
\node [draw] at (0,4) (yl) {$y.\mathsf{load}^\oplus$};
\node [draw] at (0,2) (b) {$B_\mathsf{R}^\oplus$};
\node [draw] at (0,0) (e) {$\textsc{\#End}$};

\path[->,every node/.style={font=\scriptsize}] (s) edge node[left] {$\stPop{A_\mathsf{L}}$} (a);
\path[->,every node/.style={font=\scriptsize}] (a) edge node[left] {$\one$} (xs);

\path[->,every node/.style={font=\scriptsize}] (yl) edge node[left] {$\one$} (b);
\path[->,every node/.style={font=\scriptsize}] (b) edge node[left] {$\stPush{B_\mathsf{R}}$} (e);

\path[->,every node/.style={font=\scriptsize}] (xs) edge node[above] {$\stPush{\mathsf{store}}$} (x);
\path[->,every node/.style={font=\scriptsize}] (y) edge node[below] {$\stPop{\mathsf{load}}$} (yl);

\path[->,every node/.style={font=\scriptsize}] (x) edge node[above] {$\one$} (pm);
\path[->,every node/.style={font=\scriptsize}] (pp) edge node[below] {$\one$} (y);

\path[->,every node/.style={font=\scriptsize}] (pm) edge[style=dotted] node[left] {$\stPop{\mathsf{store}}$} (psp);
\path[->,every node/.style={font=\scriptsize}] (pm) edge[style=dotted] node[right] {$\stPop{\mathsf{load}}$} (plm);
\path[->,every node/.style={font=\scriptsize}] (psp) edge[style=dotted] node[left] {$\one$} (plp);
\path[->,every node/.style={font=\scriptsize}] (plm) edge[style=dotted] node[right] {$\one$} (psm);
\path[->,every node/.style={font=\scriptsize}] (plp) edge[style=dotted] node[left] {$\stPush{\mathsf{load}}$} (pp);
\path[->,every node/.style={font=\scriptsize}] (psm) edge[style=dotted] node[right] {$\stPush{\mathsf{store}}$} (pp);

\path[->,every node/.style={font=\scriptsize}]
  (xs.south east) edge[style=dashed, bend left] node[right] {$\one$} (yl.north east);

\end{tikzpicture}\]
\end{minipage}
\begin{minipage}{0.3\textwidth}
\begin{verbatim}
{
  p =  y;
  x =  p;
 *x =  A;
  B = *y;
}
\end{verbatim}
\end{minipage}
\caption{Saturation using an implicit application of $p.\mathsf{store} \subtype p.\mathsf{load}$. The
initial constraint set was $\{y \subtype p,\quad p \subtype x,\quad A \subtype x.\mathsf{store},\quad
y.\mathsf{load} \subtype B\}$, modeling the simple program on the right. The dashed edge was
added by a saturation rule that only fires because of the lazy handling in \autoref{sat-alg-code}.
The dotted states and edges show how the graph would look if the corresponding rule from
$\Delta_\mathsf{ptr}$ were explicitly instantiated.}
\label{pointer-sat}
\end{figure*}
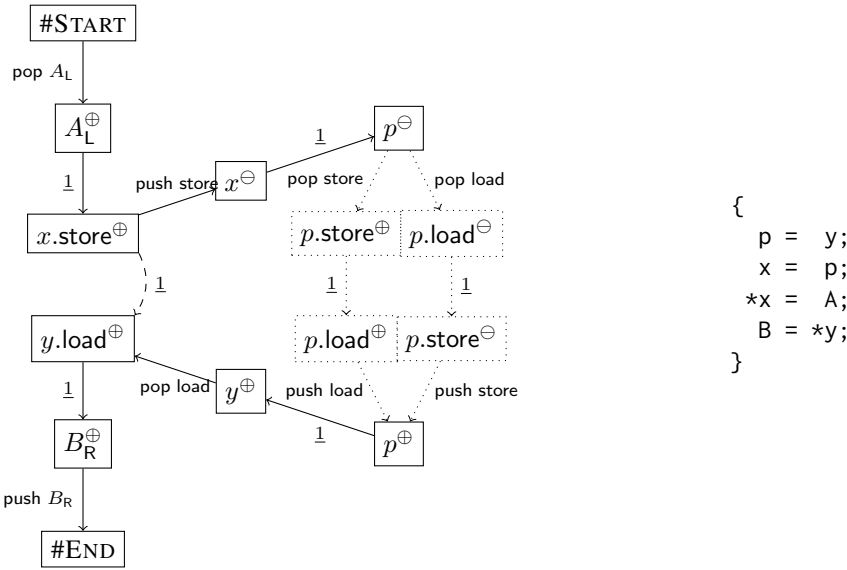

Once the saturation algorithm completes, the automaton $\mathcal{A}^\mathsf{sat}_\mathcal{C}$ has the
property that, if there is a transition sequence from $p$ to $q$ with weight {\em equivalent} to
\[\stPop{u_1} \otimes \cdots \otimes \stPop{u_n} \otimes \stPush{v_m} \otimes \cdots \otimes \stPush{v_1},\]
then there is a path from $p$ to $q$ that has, ignoring $\one$ edges, {\em exactly} the label sequence
\[\stPop{u_1}, \dots, \stPop{u_n}, \stPush{v_m}, \dots, \stPush{v_1}.\]

\subsection{Shadowing}
\label{shadowing}

The automaton $\mathcal{A}^\mathsf{sat}_\mathcal{C}$ now accepts push/pop sequences representing the
changes in the stack during any legal derivation in the pushdown system $\Delta$.  After
saturation, we can guarantee that every derivation is represented by a path which first pops
a sequence of tokens, then pushes another sequence of tokens.

Unfortunately, $\mathcal{A}^\mathsf{sat}_\mathcal{C}$ still accepts unproductive transition sequences
which push and then immediately pop token sequences.  To complete the construction, we
form an automaton $Q$ by intersecting $\mathcal{A}^\mathsf{sat}_\mathcal{C}$
with an automaton for the language of words consisting of only pops, followed by only pushes.

This final transformation yields an
automaton $Q$ with the property that for every transition sequence $s$ accepted by
$\mathcal{A}^\mathsf{sat}_\mathcal{C}$, $Q$ accepts a sequence $s'$ such that $\sem{s} = \sem{s'}$ in 
$\mathsf{StackOp}$ and $s'$ consists of a sequence of pops followed by a sequence of pushes.
 This ensures
that $Q$ only accepts the productive transition sequences in $\mathcal{A}^\mathsf{sat}_\mathcal{C}$.
Finally, we can treat $Q$ as a finite-state transducer by treating transitions labeled $\stPop{\ell}$
as reading the symbol $\ell$ from the input tape, $\stPush{\ell}$ as writing $\ell$ to the output
tape, and $\one$ as an $\varepsilon$-transition.

$Q$ can be further manipulated through determinization and/or minimization to produce a
compact transducer representing the valid transition sequences through $\mathcal{P}_\mathcal{C}$.
Taken as a whole, we have shown that $Xu \stackrel{Q}{\mapsto} Yv$ if and only if $X$ and $Y$
are interesting variables and there is an elementary derivation of $\mathcal{C} \proves X.u \subtype Y.v$.

We make use of this process in two places during type analysis: first, by computing $Q$ relative
to the type variable of a function, we get a transducer that represents all
elementary derivations of relationships between the function's inputs and outputs.
\autoref{transducer-to-pds} is used to convert the transducer $Q$ back to a pushdown
system $\mathcal{P}$, such that $Q$ describes all valid derivations in $\mathcal{P}$.
Then the rules in $\mathcal{P}$ can be interpreted as subtype constraints, resulting in a 
simplification of the constraint set relative to the formal type variables.

Second, by computing $Q$ relative to the set of type constants we obtain a transducer that
can be efficiently queried to determine which derived type variables are bound above or below
by which type constants.  This is used by the $\textsc{Solve}$ procedure in \autoref{type-infer-alg2} 
to populate lattice elements decorating the inferred sketches.

\begin{algorithm}
\caption{Converting a transducer to a pushdown system}
\label{transducer-to-pds}
\begin{algorithmic}[0]

  \Procedure{TypeScheme}{Q}
    \State{$\Delta \gets \text{new PDS}$}
    \State{$\Delta.\mathsf{states} \gets Q.\mathsf{states}$}
    \ForAll{$p \stackrel{t}{\to} q \in Q.\mathsf{transitions}$}
      \If{$t = \stPop{\ell}$}
        \State{\Call{AddPDSRule}{$\Delta, \pdsrule{p}{\ell}{q}{\varepsilon}$}}
      \Else
        \State{\Call{AddPDSRule}{$\Delta, \pdsrule{p}{\varepsilon}{q}{\ell}$}}
      \EndIf
    \EndFor
    \State{\Return{$\Delta$}}
  \EndProcedure

\end{algorithmic}
\end{algorithm}

\section{The Lattice of Sketches}
\label{sketch-appendix}

Throughout this section, we fix a lattice $(\Lambda, <:, \vee, \wedge)$ of atomic types.  We
do not assume anything about $\Lambda$ except that it should have finite height so that
infinite subtype chains eventually stabilize.
For example purposes, we will take 
$\Lambda$ to be the lattice of semantic classes depicted in \autoref{sample-lattice}.

Our initial implementation of sketches did not use the auxiliary lattice $\Lambda$.
We found that adding these decorations to the sketches helped preserve high-level types of interest to the
end user during type inference.  This allows us to
recover high-level C and Windows typedefs such as {\verb|size_t|}, {\verb|FILE|},
{\verb|HANDLE|}, and {\verb|SOCKET|} that are useful for program understanding and
reverse engineering, as noted by \citet{rewards}.

Decorations also enable a simple mechanism by which the user can extend Retypd's type system, adding
semantic purposes to the types for known functions. For example, we can extend $\Lambda$ to add seeds for
a tag {\verb|#signal-number|} attached to the less-informative {\verb|int|} parameter to
{\verb|signal()|}.
This approach also allows us to distinguish between opaque typedefs and the underlying type, as in
{\verb|HANDLE|} and {\verb|void*|}.  Since the semantics of a {\verb|HANDLE|} are quite
distinct from those
of a {\verb|void*|}, it is important to have a mechanism that can preserve the typedef name.

\begin{figure}
\[\begin{tikzpicture}[scale=0.7]
\node (top) at (0,3)  {$\top$};
\node (num) at (-1,2) {$\typename{num}$};
\node (str) at (1,2)  {$\typename{str}$};
\node (url) at (1,1)  {$\typename{url}$};
\node (bot) at (0,0)  {$\bot$};
\draw (top) -- (num) -- (bot);
\draw (top) -- (str) -- (url) -- (bot);
\end{tikzpicture}\]
\caption{The sample lattice $\Lambda$ of atomic types.}
\label{sample-lattice}
\end{figure}

\subsection{Basic Definitions}
\begin{definition}
A {\em sketch} is a regular tree with edges labeled by elements of $\Sigma$ and nodes labeled
with elements of $\Lambda$.  The set of all sketches, with $\Sigma$ and $\Lambda$ implicitly
fixed, will be denoted $\mathsf{Sk}$.  We may alternately think of a sketch as a 
prefix-closed regular language $\mathcal{L}(S) \subseteq \Sigma^*$ and a function $\nu : S \to \Lambda$
such that each fiber $\nu^{-1}(\lambda)$ is regular.
It will be convenient to write $\nu_S(w)$
for the value of $\nu$ at the node of $S$ reached by following the word $w \in \Sigma^*$.

By collapsing equal subtrees, we can represent sketches as deterministic finite state automata with
each state labeled by an element of $\Lambda$, as in \autoref{sketch-example}.
Since the regular language associated with a sketch is prefix-closed, all states of the associated automaton
are accepting.
\end{definition}

\begin{lemma}
The set of sketches $\mathsf{Sk}$ forms a lattice with meet and join operations $\sqcap, \sqcup$ defined
according to \autoref{sketch-lattice}.  $\mathsf{Sk}$ has a top element given by the sketch accepting
the language $\{\varepsilon\}$, with the single node labeled by $\top \in \Lambda$.

If $\Sigma$ is finite, $\mathsf{Sk}$ also has a bottom element accepting the language $\Sigma^*$, with
label function $\nu_\bot(w) = \bot$ when $\langle w \rangle = \oplus$, or $\nu_\bot(w) = \top$ when
$\langle w \rangle = \ominus$.

We will use $X \subsketch Y$ to denote the partial order on sketches compatible with the lattice operations,
so that $X \sqcap Y = X$ if and only if $X \subsketch Y$.
\end{lemma}

\subsubsection{Modeling Constraint Solutions with Sketches}

Sketches are our choice of entity for modeling solutions to the constraint sets of
\autoref{section-syntax}.
\begin{definition}A {\em solution} to a constraint set $\mathcal{C}$ over the
type variables $\mathcal{V}$ is a set of bindings $S : \mathcal{V} \to \mathsf{Sk}$
such that
\begin{itemize}
\item If $\overline{\kappa}$ is a type constant, $\mathcal{L}(S_\kappa) = \{\varepsilon\}$
and $\nu_{S_\kappa}(\varepsilon) = \kappa$. 
\item If $\mathcal{C} \proves \term{X.v}$ then $v \in \mathcal{L}(X)$.
\item If $\mathcal{C} \proves X.u \subtype Y.v$ then $\nu_X(u) <: \nu_Y(v)$.
\item If $\mathcal{C} \proves X.u \subtype Y.v$ then $u^{-1} S_X \subsketch v^{-1} S_Y$,
where $v^{-1} S_X$ is the sketch corresponding to the subtree reached by 
following the path $v$ from the root of $S_X$.
\end{itemize}
\end{definition}

The main utility of sketches is that they are almost a free tree model of the constraint language.  
Any constraint set $\mathcal{C}$ is satisfiable over the lattice of sketches, as long as
$\mathcal{C}$ cannot prove an impossible subtype relation in $\Lambda$.

\begin{theorem}
Suppose that $\mathcal{C}$ is a constraint set over the variables $\{\tau_i\}_{i \in I}$.
Then there exist sketches $\{S_i\}_{i \in I}$ such that $w \in S_i$ if and only if
$\mathcal{C} \proves \term{\tau_i.w}$.
\end{theorem}
\begin{proof}
The languages $\mathcal{L}(S_i)$ can be computed an algorithm that is similar in
spirit to Steensgaard's
method of almost-linear-time pointer analysis \citep{steensgaard}.
Begin by forming a graph with one
node $n(\alpha)$ for each
derived type variable appearing in $\mathcal{C}$, along with each of its prefixes.  Add a labeled
edge $n(\alpha) \stackrel{\ell}{\to} n(\alpha.\ell)$  for each derived type variable $\alpha.\ell$ to
form a graph $G$.  Now quotient $G$ by the equivalence relation $\sim$ defined by $n(\alpha) \sim
n(\beta)$ if $\alpha \subtype \beta \in \mathcal{C}$, and $n(\alpha') \sim n(\beta')$ whenever there are edges
$n(\alpha) \stackrel{\ell}{\to} n(\alpha')$ and $n(\beta) \stackrel{\ell'}{\to} n(\beta')$ in $G$
with $n(\alpha) \sim n(\beta)$ where either $\ell = \ell'$ or $\ell = \mathsf{.load}$ and $\ell' = 
\mathsf{.store}$.

The relation $\sim$ is the symmetrization of $\subtype$, with the first defining rule roughly
corresponding to $\textsc{T-InheritL}$ and $\textsc{T-InheritR}$, and the second rule corresponding to
$\textsc{S-Field}_\oplus$ and $\textsc{S-Field}_\ominus$.  The unusual condition on $\ell$ and $\ell'$ is due to
the $\textsc{S-Pointer}$ rule.

By construction, there exists a path with label sequence $u$ through $G/\!\!\sim$ starting at the equivalence class of
$\tau_i$ if and only if $\mathcal{C} \proves \term{\tau_i.u}$.  We can take this as the definition of the
language accepted by $S_i$.
\end{proof}

Working out the lattice elements that should label $S_i$ is a trickier problem; the basic idea is to use
the same pushdown system construction that appears during constraint simplification to answer queries
about which type constants are upper and lower bounds on a given derived type variable.
The computation of upper and lower lattice bounds on a derived type variable appears in \autoref{shadowing}.

\begin{algorithm}
\caption{Computing sketches from constraint sets}
\label{type-infer-alg4}
\begin{algorithmic}[0]

  \Procedure{InferShapes}{$\mathcal{C}_\text{initial}, B$}
    \State{$\mathcal{C} \gets $ \Call{Substitute}{$\mathcal{C}_\text{initial}, B$}}
    \State{$G \gets \emptyset$}
    \Comment{Compute constraint graph modulo $\sim$}
    \ForAll{$p.\ell_1 \dots \ell_n \in \mathcal{C}.\mathsf{derivedTypeVars}$}
      \For{$i \gets 1 \dots n$}
        \State{$s \gets $ \Call{FindEquivRep}{$p. \ell_1 \dots \ell_{i-1}, G$}}
        \State{$t \gets $ \Call{FindEquivRep}{$p. \ell_1 \dots \ell_{i}, G$}}
        \State{$G.\mathsf{edges} \gets G.\mathsf{edges} \cup (s, t, \ell_i)$}
      \EndFor
    \EndFor
    \ForAll{$x \subtype y \in \mathcal{C}$}
      \State{$X \gets $ \Call{FindEquivRep}{$x, G$}}
      \State{$Y \gets $ \Call{FindEquivRep}{$y, G$}}
      \State{\Call{Unify}{$X, Y, G$}}
    \EndFor
    \Repeat
      \Comment{Apply additive constraints and update $G$}
      \State{$\mathcal{C}_\text{old} \gets \mathcal{C}$}
      \ForAll{$c \in \mathcal{C}_\text{old}$ with $c = \textsc{Add}(\_)$ or $\textsc{Sub}(\_)$}
        \State{$D \gets $ \Call{ApplyAddSub}{$c, G, \mathcal{C}$}}
        \ForAll{$\delta \in D$ with $\delta = X \subtype Y$}
          \State{\Call{Unify}{$X, Y, B$}}
        \EndFor
      \EndFor
    \Until{$\mathcal{C}_\text{old} = \mathcal{C}$}
    \ForAll{$v \in \mathcal{C}.\mathsf{typeVars}$}
      \Comment{Infer initial sketches}
      \State{$S \gets \text{new Sketch}$}
      \State{$\mathcal{L}(S) \gets $ \Call{AllPathsFrom}{$v, G$}}
      \ForAll{states $w \in S$}
        \If{$\langle w \rangle = \oplus$}
          \State{$\nu_S(w) \gets \top$}
        \Else
          \State{$\nu_S(w) \gets \bot$}
        \EndIf
      \EndFor
      \State{$B[v] \gets S$}
    \EndFor
  \EndProcedure

  \Statex

  \Procedure{Unify}{$X, Y, G$}
    \Comment{Make $X \sim Y$ in $G$}
    \If{$X \neq Y$}
      \State{\Call{MakeEquiv}{$X, Y, G$}}
      \ForAll{$(X', \ell) \in G.\mathsf{outEdges}(X)$}
        \If{$(Y', \ell) \in G.\mathsf{outEdges}(Y)$ for some $Y'$}
          \State{\Call{Unify}{$X',Y',G$}}
        \EndIf
      \EndFor
    \EndIf
  \EndProcedure

\end{algorithmic}
\end{algorithm}

\subsubsection{Sketch Narrowing at Function Calls}
\label{reverse-dns}

\begin{figure}
\[\begin{tikzpicture}[->,shorten >=1pt,auto,semithick,node distance=2cm,initial text={},scale=0.7]
   \tikzstyle{every state}=[draw]

   \node[initial,state] (L)                {\scriptsize$\top$};
   \node[state]         (Lp) [right of=L]  {\scriptsize$\top$};
   \node[state]         (Lu) [right of=Lp] {\scriptsize$\typename{str}$};
   \node[state]         (N)  [below of=L]  {\scriptsize$\bot$};
   \node[state]         (Np) [right of=N]  {\scriptsize$\bot$};
   \node[state]         (Nu) [right of=Np] {\scriptsize$\typename{str}$};

   \path (L)  edge                  node {\scriptsize$\mathsf{.store}$      } (Np)
         (N)  edge                  node {\scriptsize$\mathsf{.store}$      } (Lp)
         (Lp) edge                  node {\scriptsize$.\sigma\mathsf{32@0}$ } (Lu)
         (Np) edge                  node {\scriptsize$.\sigma\mathsf{32@0}$ } (Nu)
         (L)  edge                  node {\scriptsize$\mathsf{.load}$       } (Lp)
         (N)  edge                  node {\scriptsize$\mathsf{.load}$       } (Np)
         (Lp) edge [bend right=60]  node [swap] {\scriptsize$.\sigma\mathsf{32@4}$ } (L)
         (Np) edge [bend left=60]   node {\scriptsize$.\sigma\mathsf{32@4}$ } (N);
\end{tikzpicture}\]
\caption{Sketch representing a linked list of strings $\texttt{struct~ LL ~\{ ~str~ s; ~struct ~LL* ~a;~ \}*}$.}
\label{sketch-example}
\end{figure}

\begin{figure}
\[\begin{tikzpicture}[->,shorten >=1pt,auto,semithick,node distance=2cm,initial text={},scale=0.7]
  \tikzstyle{every state}=[draw]

  \node[initial,state] (L)                {\scriptsize$\top$};
  \node[state]         (Lp) [right of=L]  {\scriptsize$\bot$};
  \node[state]         (Lu) [right of=Lp] {\scriptsize$\typename{url}$};

  \path (Lp) edge                  node {\scriptsize$.\sigma\mathsf{32@0}$ } (Lu)
        (L)  edge                  node {\scriptsize$\mathsf{.store}$      } (Lp);
\end{tikzpicture}\]
\caption{Sketch for $Y$ representing \tt{\_Out\_ url * u;}}
\end{figure}

It was noted in \autoref{deduction-rule-sec} that the rule $\textsc{T-InheritR}$ leads to a
system with structural typing: any two types in a subtype relation must have the same fields.
In the language of sketches, this means that if two type variables are in a subtype relation
then the corresponding sketches accept exactly the same languages.  Superficially, this seems
problematic for modeling typecasts that narrow a pointed-to object as motivated by
the idioms in \autoref{ptr-to-member}.

The missing piece that allows us to effectively narrow objects is instantiation of callee type
schemes at a callsite.  To demonstrate how polymorphism enables narrowing, consider the 
example type scheme $\scheme{F}{\mathcal{C}}{F}$ from \autoref{type-recovery-example}.
The function {\verb|close_last|} can be invoked by providing any actual-in type $\alpha$
such that $\alpha \subtype F.\mathsf{in}_{\mathsf{stack}0}$; in particular, $\alpha$ can have
{\em more} capabilities than $F.\mathsf{in}_{\mathsf{stack}0}$ itself. That is, we can
pass a more constrainted (``more capable'') type as an actual-in to a function that expected
a less constrained input.  In this way, we recover aspects of the physical, nonstructural
subtyping utilized by many C and C++ programs via the pointer-to-member or pointer-to-substructure
idioms described in \autoref{ptr-to-member}.

\begin{example}
Suppose we have a reverse DNS lookup function {\verb|reverse_dns|} with C type signature
  {\verb|void reverse_dns|}\\{\verb|(num addr,url* result)|}.  Furthermore, assume that
the implementation of {\verb|reverse_dns|} works by writing the resulting URL
to the location pointed to by {\verb|result|}, yielding a constraint of the form
$\typename{url} \subtype \texttt{result}.\mathsf{store}.\sigma\mathsf{32@0}$.  So
{\verb|reverse_dns|} will have an inferred type scheme of the form
\[ \forall \alpha, \beta .
(\alpha \subtype \typename{num}, \typename{url} \subtype \beta.\mathsf{store}.\sigma\mathsf{32@0})
 \Rightarrow \alpha \times \beta \to \mathsf{void}\]
Now suppose we have the structure {\verb|LL|} of \autoref{sketch-example}
representing a linked list of strings, with instance {\verb|LL * myList|}.
Can we safely invoke {\verb|reverse_dns(addr, (url*) myList)|}?

Intuitively, we can see that this should be possible: since the linked list's payload is
in its first field, a value of type {\verb|LL*|} also looks like a value of type
{\verb|str*|}.  Furthermore, {\verb|myList|} is not {\verb|const|},
so it can be used to store the function's output.

Is this intuition borne out by Retypd's type system?  The answer is yes, though it takes
a bit of work to see why.  Let us write $S_\beta$ for the instantiated sketch of $\beta$
at the callsite of {\verb|reverse_dns|} in question; the constraints $\mathcal{C}$
require that $S_\beta$ satisfy
$.\mathsf{store}.\sigma\mathsf{32@0} \in \mathcal{L}(S_\beta)$ and
$\typename{url} <: \nu_{S_\beta}(.\mathsf{store}.\sigma\mathsf{32@0})$.  The actual-in has
the type sketch $S_X$ seen in \autoref{sketch-example}, and the copy from actual-in to
formal-in will generate the constraint $X \subtype \beta$. 

Since we have already satisfied the two constraints on $S_\beta$ coming from its use in {\verb|reverse_dns|},
we can freely add other words to $\mathcal{L}(S_\beta)$, and can freely set their node labels to almost
any value we please subject only to the constraint $S_X \subsketch S_\beta$.
This is simple enough: we just add every word in $\mathcal{L}(S_X)$ to $\mathcal{L}(S_\beta)$.
If $w$ is a word accepted by $\mathcal{L}(S_X)$ that must be added to $\mathcal{L}(S_\beta)$, we will define
$\nu_{S_\beta}(w) := \nu_{S_X}(w)$.  Thus, both the shape and the node labeling on $S_X$ and $S_\beta$ match, with
one possible exception: we must check that $X$ has a nested field $.\mathsf{store}.\sigma\mathsf{32@0}$, and
that the node labels satisfy the relation 
\[\nu_{S_\beta}(.\mathsf{store}.\sigma\mathsf{32@0}) <: \nu_{S_X}(.\mathsf{store}.\sigma\mathsf{32@0})\] 
since $.\mathsf{store}.\sigma\mathsf{32@0}$ is contravariant and $S_X \subsketch S_\beta$.
$X$ does indeed have the required field, and $\typename{url} <: \typename{str}$; the
function invocation is judged to be type-safe.
\end{example}
 
\section{Additional Algorithms}
This appendix holds a few algorithms referenced in the main text and other appendices.
\label{the-algorithms}

\begin{algorithm}
\caption{Type scheme inference}
\label{type-infer-alg}
\begin{algorithmic}[0]

  \Procedure{InferProcTypes}{$\mathsf{CallGraph}$}
    \State{$T \gets \emptyset$}
    \Comment{$T$ is a map from procedure to type scheme.}
    \ForAll{$S \in$ \Call{Postorder}{$\mathsf{CallGraph.sccs}$}}
      \State{$\mathcal{C} \gets \emptyset$}
      \ForAll{$P \in S$}
        \State{$T\left[P\right] \gets \emptyset$}
      \EndFor
      \ForAll{$P \in S$}
        \State{$\mathcal{C} \gets \mathcal{C} ~\cup$ \Call{Constraints}{$P, T$}}
      \EndFor
      \State{$\mathcal{C} \gets $ \Call{InferShapes}{$\mathcal{C}, \emptyset$}}
      \ForAll{$P \in S$}
        \State{$\mathcal{V} \gets P.\mathsf{formalIns} \cup P.\mathsf{formalOuts}$}
        \State{$Q \gets $\Call{Transducer}{$\mathcal{C}, \mathcal{V} \cup \overline{\Lambda}$}}
        \State{$T\left[P\right] \gets $\Call{TypeScheme}{$Q$}}
      \EndFor
    \EndFor
  \EndProcedure

  \Statex

  \Procedure{Constraints}{$P, T$}
    \State{$\mathcal{C} \gets \emptyset$}
    \ForAll{$i \in P.\mathsf{instructions}$}
      \State{$C \gets C ~\cup$ \Call{AbstractInterp}{$\mathsf{TypeInterp}, i$}}
      \If{$i$ calls $Q$}
        \State{$\mathcal{C} \gets \mathcal{C} ~\cup$ \Call{Instantiate}{$T[Q], i$}}
      \EndIf
    \EndFor
    \State{\Return{$\mathcal{C}$}}
  \EndProcedure

\end{algorithmic}
\end{algorithm}

\begin{algorithm}
\caption{C type inference}
\label{type-infer-alg2}
\begin{algorithmic}[0]

  \Procedure{InferTypes}{$\mathsf{CallGraph}, T$}
    \State{$B \gets \emptyset$}
    \Comment{$B$ is a map from type variable to sketch.}
    \ForAll{$S \in$ \Call{ReversePostorder}{$\mathsf{CallGraph.sccs}$}}
      \State{$\mathcal{C} \gets \emptyset$}
      \ForAll{$P \in S$}
        \State{$T\left[P\right] \gets \emptyset$}
      \EndFor
      \ForAll{$P \in S$}
        \State{$\mathcal{C}_\partial \gets T[P]$}
        \State{\Call{Solve}{$\mathcal{C}_\partial, B$}}
        \State{\Call{RefineParameters}{$P, B$}}
        \State{$\mathcal{C} \gets$ \Call{Constraints}{$P, T$}}
        \State{\Call{Solve}{$\mathcal{C}, B$}}
      \EndFor
    \EndFor
    \State{$A \gets \emptyset$}
    \ForAll{$x \in B.\mathsf{keys}$}
      \State{$A[x] \gets $ \Call{SketchToAppxCType}{$B[x]$}}
    \EndFor
    \State{\Return{$A$}}
  \EndProcedure

  \Statex

  \Procedure{Solve}{$\mathcal{C}, B$}
    \State{$\mathcal{C} \gets $ \Call{InferShapes}{$\mathcal{C}, B$}}
    \State{$Q \gets $\Call{Transducer}{$\mathcal{C}, \overline{\Lambda}$}}
    \ForAll{$\lambda \in \Lambda$}
      \ForAll{$X u$ such that $\lambda \stackrel{Q}{\mapsto} X u$}
        \State{$\nu_{B[X]}(u) \gets \nu_{B[X]}(u) \vee \lambda$}
      \EndFor
      \ForAll{$X u$ such that $X u \stackrel{Q}{\mapsto} \lambda$}
        \State{$\nu_{B[X]}(u) \gets \nu_{B[X]}(u) \wedge \lambda$}
      \EndFor
    \EndFor
  \EndProcedure

\end{algorithmic}
\end{algorithm}

\begin{algorithm}
\caption{Procedure specialization}
\label{type-infer-alg3}
\begin{algorithmic}[0]

  \Procedure{RefineParameters}{$P, B$}
    \ForAll{$i \in P.\mathsf{formalIns}$}
      \State{$\lambda \gets \top$}
         \ForAll{$a \in P.\mathsf{actualIns}(i)$}
            \State{$\lambda \gets \lambda \sqcup B[a]$}
         \EndFor
      \State{$B[i] \gets B[i] \sqcap \lambda$}
    \EndFor
    \ForAll{$o \in P.\mathsf{formalOuts}$}
      \State{$\lambda \gets \bot$}
         \ForAll{$a \in P.\mathsf{actualOuts}(o)$}
            \State{$\lambda \gets \lambda \sqcap B[a]$}
         \EndFor
      \State{$B[o] \gets B[o] \sqcup \lambda$}
    \EndFor
  \EndProcedure

\end{algorithmic}
\end{algorithm}

\begin{figure}
\centering
\begin{minipage}{.9\linewidth}
\begin{tabular}{c}
\begin{minipage}{\linewidth}
     \vspace{0.05in}
        \centering
     \emph{Language}
     \[\mathcal{L}(X \sqcap Y) = \mathcal{L}(X) \cup \mathcal{L}(Y)\]
     \[\mathcal{L}(X \sqcup Y) = \mathcal{L}(X) \cap \mathcal{L}(Y)\]
      \vspace{-0.1in}
\end{minipage} \\
\begin{minipage}{\linewidth}
     \vspace{0.05in}
        \centering
  \emph{Node labels}
      \[\nu_{X \sqcap Y}(w) = \begin{cases}
          \nu_X(w) \wedge \nu_Y(w) & \text{if } \langle w \rangle = \oplus \\
          \nu_X(w) \vee \nu_Y(w) & \text{if } \langle w \rangle = \ominus
          \end{cases} \]
      \[\nu_{X \sqcup Y}(w) =
        \begin{cases}
          \nu_X(w) & \text{if } w \in \mathcal{L}(X) \setminus \mathcal{L}(Y) \\
          \nu_Y(w) & \text{if } w \in \mathcal{L}(Y) \setminus \mathcal{L}(X) \\
          \nu_X(w) \vee \nu_Y(w) & \text{if } w \in \mathcal{L}(X) \cap \mathcal{L}(Y),\\
          & \quad \langle w \rangle = \oplus \\
          \nu_X(w) \wedge \nu_Y(w) & \text{if } w \in \mathcal{L}(X) \cap \mathcal{L}(Y),\\
          & \quad \langle w \rangle = \ominus
        \end{cases}\]
        \vspace{0.05in}
\end{minipage} \\
\end{tabular}
\end{minipage}
\caption{Lattice operations on the set of sketches.}
\label{sketch-lattice}
\end{figure}

\section{Other C Type Resolution Policies}

\begin{example}
The initial type-simplification stage results in types that are as general as possible.
Often, this means that types are found to be more general than is strictly helpful to
a (human) observer. A policy called $\textsc{RefineParameters}$ is used to specialize
type schemes to the most {\em specific} scheme that is compatible with all uses.
For example, a C++ object may include a getter function with a highly polymorphic
type scheme, since it could operate equally well on any structure with a field of the right type at the
right offset.  But we expect that in every calling context, the getter will be called on a specific 
object type (or perhaps its derived types).  By specializing the function signature, we make use of
contextual clues in exchange for generality before presenting a final C type to the user.
\end{example}

\begin{example}
  Suppose we have a C++ class
  
\begin{code}
class MyFile
{
 public:
   char * filename() const {
     return m_filename;
   }
 private:
   FILE * m_handle;
   char * m_filename;
};
\end{code}

\noindent
In a 32-bit binary, the implementation of {\verb|MyFile::filename|} (if not inlined) will be
roughly equivalent to the C code
\begin{code}
typedef int32_t dword;
dword get_filename(const void * this)
{
   char * raw_ptr = (char*) this;
   dword * field_ptr =
       (dword*) (raw_ptr + 4);
   return *field_ptr;
}
\end{code}

\noindent
Accordingly, we would expect the most-general inferred type scheme for {\verb|MyFile::filename|} to be
\[\forall \alpha, \beta . (\beta \subtype \texttt{dword},
\alpha.\mathsf{load}.\sigma\mathsf{32@4} \subtype \beta) \Rightarrow \alpha \to \beta\]
indicating that {\verb|get_filename|} will accept a pointer to anything which has a value of some 32-bit
type $\beta$ at offset 4, and will return a value of that same type. If the function is truly used
polymorphically, this is exactly the kind of precision that we wanted our type system to maintain.

But in the more common case, {\verb|get_filename|} will only be called with values where
{\verb|this|} has type {\verb|MyFile*|} (or perhaps a subtype, if we include inheritance).
If every callsite to
{\verb|get_filename|} passes it a pointer to {\verb|MyFile*|}, it may be best to specialize
the type of {\verb|get_function|} to the monomorphic type
\[\mathsf{get\_filename} : \texttt{const~MyFile*} \to \texttt{char*}\]
\end{example}

The function $\textsc{RefineParameters}$ in \autoref{type-infer-alg3} is used to specialize each
function's type {\em just enough} to match how the function is actually used in a program, at the cost of reduced generality.

\begin{example}
A useful but less sound heuristic is represented by the {\verb|reroll|} policy for
handling types which look like unrolled recursive types:
\begin{lstlisting}[mathescape]
reroll(x):
 if there are u and $\ell$ with x.$\ell$u = x.$\ell$,
      and sketch(x) $\subsketch$ sketch(x.$\ell$):
   replace x with x.$\ell$
 else:
   policy does not apply
\end{lstlisting}
In practice, we often need to add other guards which inspect the shape of $x$ to determine
if the application of reroll appears to be appropriate or not. For example, we may require
$x$ to have at least one field other than $\ell$ to help distinguish a pointer-to-linked-list
from a pointer-to-pointer-to-linked-list.
\end{example}
 
\newpage
\section{Details of the \autoref{type-recovery-example} Example}
The results of constraint generation for the example program in \autoref{type-recovery-example} appears in \autoref{generated-constraints}.  The constraint-simplification algorithm
builds the automaton $Q$ (\autoref{q-dot}) to recognize the simplified entailment closure
of the constraint set.  $Q$ recognizes exactly the input/output pairs of the form
\[\left(\mathsf{close\_last}.\mathsf{in}_{\mathsf{stack}0}(.\mathsf{load}.\sigma\mathsf{32@0})^*
.\mathsf{load}.\sigma\mathsf{32@4}, (\mathrm{int} ~|~ \mathrm{\#FileDescriptor})\right)\]
and
\[\left((\mathrm{int} ~|~ \mathrm{\#SuccessZ}), \mathsf{close\_last}.\mathsf{out}_{\mathsf{eax}}\right)\]
To generate the simplified constraint set, a type variable $\tau$ is synthesized for the single internal state in $Q$.  The path leading from the start state to $\tau$ generates the constraint
\[\mathsf{close\_last}.\mathsf{in}_{\mathsf{stack}0} \subtype \tau\]
The loop transition generates
\[\tau.\mathsf{load}.\sigma\mathsf{32@0} \subtype \tau\]
and the two transitions out of $\tau$ generate
\begin{align*}
\tau.\mathsf{load}.\sigma\mathsf{32@4} &\subtype \mathrm{int} \\
\tau.\mathsf{load}.\sigma\mathsf{32@4} &\subtype \mathrm{\#FileDescriptor}
\end{align*}
Finally, the two remaining transitions from start to end generate
\begin{align*}
  \mathrm{int} &\subtype \mathsf{close\_last}.\mathsf{out}_{\mathsf{eax}} \\
  \mathrm{\#SuccessZ} &\subtype \mathsf{close\_last}.\mathsf{out}_{\mathsf{eax}}
\end{align*}
To generate the simplified constraint set, we gather up these constraints (applying
some lattice operations to combine inequalities that only differ by a lattice constant)
and close over the introduced $\tau$ by introducing an $\exists \tau$ quantifier.
The result is the constraint set of \autoref{type-recovery-example}.

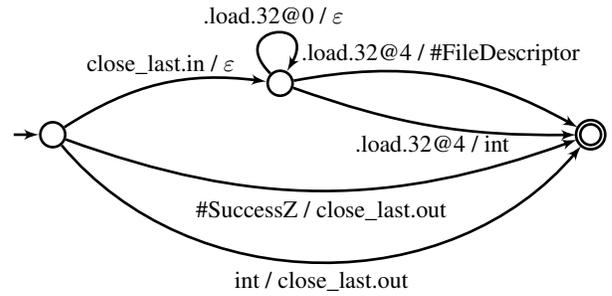
\begin{figure}[b]
  \[
  \begin{tikzpicture}[>=latex',line join=bevel,scale=0.4]
  \pgfsetlinewidth{1bp}
  \pgfsetcolor{black}
  \node[shape=circle,draw] (s) at (37.942bp,19.588bp) {};
  \node[shape=circle,draw] (t) at (252.16bp,68.075bp) {};
  \node[shape=circle,draw,accepting] (e) at (544.7bp,19.588bp) {};

  \draw [->] (1.1549bp,19.588bp) .. controls (2.6725bp,19.588bp) and (15.097bp,19.588bp)  .. (s);

  \path [->] (s) edge [bend left=20]  node [above] {\small{close\_last.in / $\varepsilon$}} (t);
  \path [->] (t) edge [bend left=20]  node [above] {\small{.load.32@4 / \#FileDescriptor}} (e);
  \path [->] (t) edge [bend right=10] node [below] {\small{.load.32@4 / int}} (e);
  \path [->] (s) edge [bend right=20] node [below] {\small{\#SuccessZ / close\_last.out}} (e);
  \path [->] (s) edge [bend right=50] node [below] {\small{int / close\_last.out}} (e);
  \path [->] (t) edge [loop above,in=60,out=135,looseness=10] node [above] {\small{.load.32@0 / $\varepsilon$}} (t);
  \begin{scope}
  \definecolor{strokecol}{rgb}{0.0,0.0,0.0};
  \pgfsetstrokecolor{strokecol}
\end{scope}
  \begin{scope}
  \definecolor{strokecol}{rgb}{0.0,0.0,0.0};
  \pgfsetstrokecolor{strokecol}
\end{scope}
  \begin{scope}
  \definecolor{strokecol}{rgb}{0.0,0.0,0.0};
  \pgfsetstrokecolor{strokecol}
\end{scope}
  \end{tikzpicture}
  \]
  \caption{The automaton $Q$ for the constraint system in \autoref{type-recovery-example}.}
  \label{q-dot}
\end{figure}
 
\begin{figure*}
  \hrule
\begin{code}[frame=none]
_text:08048420  close_last proc near 
_text:08048420  close_last:
_text:08048420          mov     edx,dword [esp+fd]
                           AR_close_last_INITIAL[4:7] <: EDX_8048420_close_last[0:3]
                           close_last.in@stack0 <: AR_close_last_INITIAL[4:7]
                           EAX_804843F_close_last[0:3] <: close_last.out@eax
_text:08048424          jmp     loc_8048432
_text:08048426          db 141, 118, 0, 141, 188, 39
_text:0804842C          times 4 db 0
_text:08048430  
_text:08048430  loc_8048430:
_text:08048430          mov     edx,eax
                           EAX_8048432_close_last[0:3] <: EDX_8048430_close_last[0:3]
_text:08048432  
_text:08048432  loc_8048432:
_text:08048432          mov     eax,dword [edx]
                           EDX_8048420_close_last[0:3] <: unknown_loc_106
                           EDX_8048430_close_last[0:3] <: unknown_loc_106
                           unknown_loc_106.load.32@0 <: EAX_8048432_close_last[0:3]
_text:08048434          test    eax,eax
_text:08048436          jnz     loc_8048430
_text:08048438          mov     eax,dword [edx+4]
                           EDX_8048420_close_last[0:3] <: unknown_loc_111
                           EDX_8048430_close_last[0:3] <: unknown_loc_111
                           unknown_loc_111.load.32@4 <: EAX_8048438_close_last[0:3]
_text:0804843B          mov     dword [esp+fd],eax
                           EAX_8048438_close_last[0:3] <: AR_close_last_804843B[4:7]
_text:0804843F          jmp     __thunk_.close
                           AR_close_last_804843B[4:7] <: close:0x804843F.in@stack0
                           close:0x804843F.in@stack0 <: #FileDescriptor
                           close:0x804843F.in@stack0 <: int
                           close:0x804843F.out@eax <: EAX_804843F_close_last[0:3]
                           int <: close:0x804843F.out@eax
_text:08048443  
_text:08048443  close_last endp
\end{code}
\caption{The constraints obtained by abstract interpretation of the example code
in \autoref{type-recovery-example}.}
\label{generated-constraints}
\end{figure*}

\end{document}